\RequirePackage{fix-cm}
\documentclass[smallextended]{svjour3}       
\smartqed  
\usepackage{graphicx}
\usepackage{amsfonts}
\usepackage{amsmath}
\makeatletter
\let\proof\@undefined
\let\endproof\@undefined
\makeatother

\usepackage{amsthm}
\usepackage{epstopdf}
\usepackage{color}
\usepackage{subfigure}

\newcommand{\cX}{{\mathcal X}}

\newcommand{\cC}{{\mathcal C}}

\newcommand{\cS}{{\mathcal S}}
\newtheorem{observation}{Observation}

\newcommand{\cs}[1]{\textcolor{black}{#1}}

\newtheorem*{prop1bis}{Proposition 1 (extended)}
\newtheorem*{prop2bis}{Proposition 2 (extended)}
\newtheorem*{lemma4bis}{Lemma 4 (extended)}
\newtheorem*{lemma7bis}{Lemma 7 (extended)}
\newtheorem*{obs5bis}{Observation 5 (extended)}

\usepackage[linesnumbered,ruled,vlined,english]{algorithm2e}

\begin{document}

\title{Constructing minimal phylogenetic networks from softwired clusters is fixed parameter tractable
}


\author{Steven Kelk         \and
        Celine Scornavacca 
}

\institute{S. Kelk \at
              Department of Knowledge Engineering (DKE), Maastricht University, P.O. Box 616, 6200 MD Maastricht, The Netherlands.  \\
              Tel.: +31 (0)43 38 82019\\
              Fax: +31 (0)43 38 84910 \\
              \email{steven.kelk@maastrichtuniversity.nl}           
           \and
           C. Scornavacca \at
              Center for Bioinformatics (ZBIT), T\"ubingen University, Sand 14, 72076 T\"ubingen, Germany\\
              \email{scornava@informatik.uni-tuebingen.de} 
}

\titlerunning{Constructing minimal softwired networks is fixed parameter tractable}

\date{Received: date / Accepted: date}

\maketitle

\begin{abstract}
Here we show that, given a set of clusters $\cC$ on a set of taxa $\cX$, where $|\cX|=n$, it is possible to determine in
time $f(k) \cdot poly(n)$ whether there exists a level-$\leq k$ network (i.e. a network where each biconnected component has reticulation number at most $k$) that represents all the clusters in $\cC$ in the softwired sense, and if so to construct such a network. This extends a polynomial time result from \cite{elusiveness}.
By generalizing the concept of ``level-$k$ generator'' to general networks, we then extend this fixed parameter tractability result to the problem where $k$ refers not to the level but to the  reticulation number of the whole network. 

\end{abstract}

\section{Introduction}

\subsection{\emph{Phylogenetic networks and softwired clusters}}

The traditional model for representing the evolution of a set  of species ${\cX}$ (or, more abstractly, a set of \emph{taxa}) is the \emph{rooted phylogenetic tree} \cite{SempleSteel2003,MathEvPhyl,reconstructingevolution}. Essentially, this is a singly-rooted tree where the leaves are bijectively labelled by ${\cX}$ and the edges are directed away from the root. In recent years there has been a growing interest in extending this model to also incorporate non-treelike evolutionary phenomena such as hybridizations, recombinations and horizontal gene transfers. This has subsequently stimulated research into \emph{rooted phylogenetic networks} which generalize rooted phylogenetic trees by
also permitting nodes with indegree two or higher, known as \emph{reticulation} nodes, or simply \emph{reticulations}. For detailed background information on phylogenetic networks we refer the reader to
\cite{HusonRuppScornavacca10,Nakhleh2009ProbSolv,Semple2007,husonetalgalled2009,twotrees,surveycombinatorial2011}. Figure \ref{fig:bcc}
shows an example of a rooted phylogenetic network.

\begin{figure}[h]
  \centering
  \includegraphics[scale=.2]{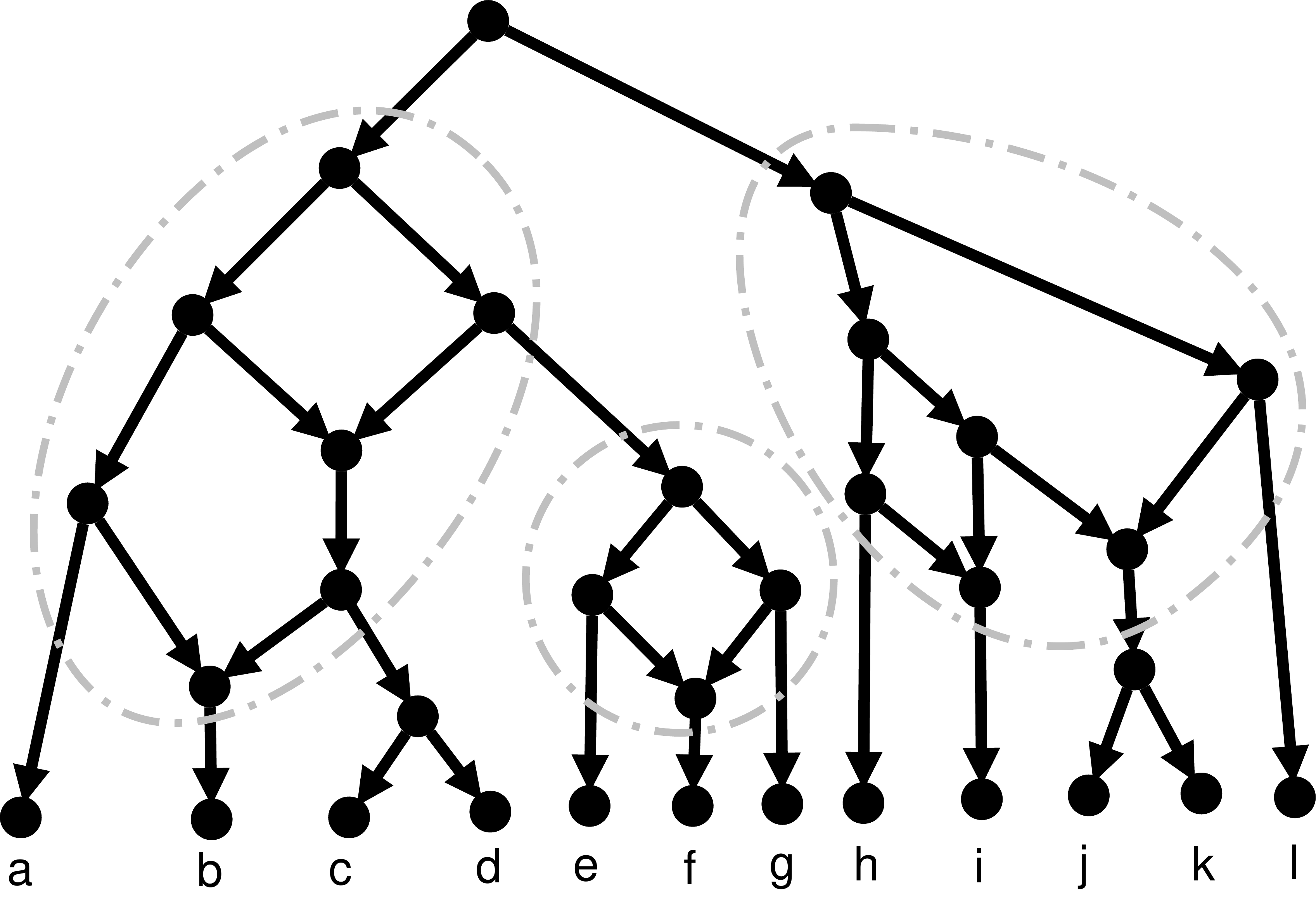}
  \caption{Example of a phylogenetic network with five reticulations. The encircled subgraphs form its biconnected components, also known as its ``tangles''. This binary network has level equal to 2 since each biconnected 
component contains at most two reticulations.}
  \label{fig:bcc}
\end{figure}

We are interested in the following biologically-motivated optimization problem. We are given a set ${\cC}$ of \emph{clusters} on ${\cX}$, where a cluster is
simply a strict subset of ${\cX}$. We wish to construct a phylogenetic network that ``represents'' all the clusters in ${\cC}$ such that the amount of reticulation in the network is ``minimized''. There are several different definitions of ``represents'' and ``minimized''
present in the literature. In this article we will consider only the \emph{softwired} definition of ``represents'' \cite{husonetalgalled2009,cass,HusonRuppScornavacca10,surveycombinatorial2011}. Most of our formal definitions will be
deferred to the preliminaries. Nevertheless, it is helpful to already formally state that a rooted phylogenetic tree $T$ on ${\cX}$ represents a cluster $C \subset {\cX}$ if $T$
contains an edge $(u,v)$ such that $C$ is exactly equal to the subset of ${\cX}$ reachable from $v$ by directed paths. A phylogenetic network $N$ on ${\cX}$,
on the other hand, represents a cluster $C \subset {\cX}$  in the softwired sense if there exists \emph{some} rooted phylogenetic tree $T$ on ${\cX}$ such that $T$ represents $C$ and
$T$ is topologically embedded inside $N$. Regarding ``minimized'', we consider two closely related, but subtly different, variants of minimality. The first variant,
\emph{reticulation number minimization}, aims at minimizing the \emph{total} number of reticulation nodes in the network\footnote{This is the definition when all reticulation
vertices have indegree-2, for more general networks reticulation number is defined slightly differently. See the
Preliminaries for more information.}. The second, less well-known variant, \emph{level minimization} \cite{JanssonSung2006,JanssonEtAl2006,lev2TCBB,reflections,tohabib2009}, asks us to minimize
the maximum number of reticulation nodes contained in any ``tangled'' region of the network, which essentially correspond to the non-trivial biconnected components of the underlying
undirected graph (see Figure \ref{fig:bcc}). The reticulation number is a global optimality criterion, while the level 
is a local optimality criterion. In general minimizing for one variant does not induce minimum solutions for the other variant (see e.g. Figure 3 of \cite{husonetalgalled2009}), although the algorithmic techniques used to tackle these problems are often related \cite{elusiveness}.

Both these problems are NP-hard and APX-hard \cite{bordewich,twotrees}. This raises the natural question: is it NP-hard to minimize the reticulation number (respectively, the level) if the number of
reticulation nodes in the network (respectively, per tangled region) is fixed? Prior to this article there were only partial answers known to these questions. In \cite{elusiveness}
it was proven that level-minimization is polynomial-time solvable if the level is fixed. A striking aspect of this proof is that the running time of the algorithm is
only polynomial time in a highly theoretical sense: it is too high to be of any practical interest. This exorbitant running time has two causes. Firstly, the exhaustive enumeration of all \emph{generators} \cite{lev2TCBB}, essentially
the set of all possible underlying topologies of a network if the taxa are ignored. Secondly, after determining the correct generator, a second wave of exhaustive enumeration determines
where a critical subset of ${\cX}$ should be located within the network, after which all remaining elements of ${\cX}$ can easily be added without much computational effort.

The question of whether a corresponding positive result would hold for reticulation
number minimization was left open, although the emergence of several partial results and practically efficient algorithms \cite{husonetalgalled2009,elusiveness} suggested that this might well be the case. Furthermore, it was not obvious how the algorithm from \cite{elusiveness} could be adapted to yield a fixed parameter tractable algorithm for level minimization -- where
the parameter is the level of the network $k$ -- since $k$ appears as an exponent of $|\cX|$ in the running time of the algorithm. 
(We refer to \cite{Flum2006,niedermeier2006,downey1999,Gramm2008} for an introduction to fixed parameter
tractability). Curiously, the main problem is not the enumeration of the generators, because the number of generators is independent of
$|\cX|$ \cite{Gambette2009structure}, but the allocation of the critical initial subset of taxa to their correct location in the network.

In this article we settle all these questions by proving for the first time that both level minimization and reticulation number minimization are fixed parameter tractable (where, in the
case of reticulation number minimization, the parameter is the reticulation number of the whole network). 
We give one algorithm for level minimization and one algorithm for reticulation minimization, although
the two algorithms have a large common core. The algorithms again rely heavily on generators, which we extend here to also be useful in the context
of reticulation number minimization; generators had hitherto only appeared in the level minimization literature. In both algorithms the major non-triviality is
showing how the network structure can still be adequately recovered if the parameter is no longer allowed to
appear in the exponent of $|\cX|$ as it was in \cite{elusiveness}.

\subsection{\emph{Beyond softwired clusters: the wider context}}

We believe that this approach is significant beyond the softwired cluster literature. Other articles discuss the problem of constructing rooted phylogenetic
networks not by combining clusters but by combining triplets \cite{simplicityAlgorithmica,reflections}, characters \cite{gusfielddecomp2007,gusfield2,WuG08,myers2003} or entire phylogenetic trees into a network. These models are in general mutually distinct although they do have a significant
common overlap which reaches its peak in the case of data derived from \emph{two} phylogenetic trees. To see this, note that if one takes the union of clusters represented
by a set of two or more phylogenetic trees, then the reticulation number (or level) required to represent these clusters is in general less than or equal to the reticulation number (or
level) required to topologically embed the trees themselves in the network, and this inequality is often strict.  However, in the case of a set comprising \emph{exactly} two trees the inequality becomes
equality \cite{twotrees}. Hence for data obtained from two trees one could solve the reticulation number minimization and level minimization problems for clusters by using algorithms
developed for the problem of topologically embedding the trees themselves into a network. These algorithms are highly efficient and fixed parameter tractable in
a practical, as opposed to solely theoretical sense \cite{bordewich2,sempbordfpt2007,quantifyingreticulation,whiddenWABI}. However, these tree algorithms do not help us with more general cluster sets, because for more than two trees the optima
of the cluster and tree models start to diverge. Indeed, the cluster model often saves reticulations with respect to the tree model by weakening the concept of ``above'' and
``below'' in the network, which is exactly why the input tree topologies do not generally survive if one atomizes them into their constituent clusters \cite{twotrees}. Moreover, the literature on
embedding three or more trees into a network is not yet mature, with articles restricting themselves to preliminary explorations \cite{pirnISMB2010,huynh}.  It therefore seems plausible that the generator approach might be adapted to the tree model (or
the other constructive methods mentioned) to yield a unified technique for producing positive complexity results for reticulation number minimization and level
minimization, even in the case of many input trees (or data obtained from many input trees).

\section{Preliminaries}

Consider a set of taxa~$\mathcal{X}$, {where $|\cX|=n$.} A \emph{rooted phylogenetic network} (on~$\mathcal{X}$), henceforth \emph{network}, is a directed acyclic graph with a single node with indegree zero 
(the \emph{root}), no nodes with both indegree and outdegree equal to 1, and nodes with outdegree zero (the \emph{leaves}) bijectively labeled by~$\mathcal{X}$. {In this article we usually identify the leaves with $\cX$}.
The indegree of a 
node~$v$ is denoted~$\delta^-(v)$ and~$v$ is called a \emph{reticulation} if~$\delta^-(v)\geq 2$, {otherwise $v$ is a \emph{tree} node}. An edge~$(u,v)$ is called a 
\emph{reticulation edge} if its target node
$v$ is a reticulation and is called a \emph{tree edge} otherwise.
 When counting reticulations in a network, we count reticulations with more than two incoming 
edges more than once because, biologically, these reticulations represent several reticulate evolutionary events. Therefore, we formally define the \emph{reticulation number} of a 
network~$N=(V,E)$ as \[r(N) = \sum_{\substack{v\in V: \delta^-(v)>0}}(\delta^-(v)-1) = |E| - |V| + 1 \enspace.\]

A \emph{rooted phylogenetic tree} on $\mathcal{X}$, henceforth \emph{tree}, is simply a network that has reticulation number zero. 
We say that a network $N$ on $\mathcal{X}$ \emph{displays}
a tree $T$ if~$T$ can be obtained from $N$ by performing a series of node and edge deletions and eventually by suppressing nodes with both indegree and outdegree equal to 1, see
Figure \ref{fig:treesclusters} for an example. We assume without loss of generality that each reticulation has outdegree at least one. Consequently, each leaf has indegree one. We say that a network is \emph{binary} if
every reticulation node has indegree 2 and outdegree 1 and every tree node that is not a leaf has outdegree 2.

 Proper subsets of~$\mathcal{X}$ are called \emph{clusters}, and a cluster $C$ is a \emph{singleton} if $|C|=1$. We say that an edge $(u,v)$ of 
a tree \emph{represents} a cluster $C \subset \cX$ if $C$ is the set of leaf descendants of $v$. A tree $T$ represents a cluster $C$ if it contains an 
edge that represents $C$. It is well-known that the set of clusters represented by a tree is a 
laminar family, often called a \emph{hierarchy} in the phylogenetics literature, and uniquely defines that tree. We say that  $N$ represents $C$ ``in the softwired sense'' if $N$ displays some tree $T$ on $\cX$ such that $T$ 
represents $C$, see Figures \ref{fig:treesclusters} and \ref{fig:example}. In this article we only consider the softwired notion of cluster representation and henceforth assume this implicitly\footnote{Alternatively, we say that a network N represents a cluster $C \subset \mathcal{X}$ ``in the hardwired sense" if there exists a tree edge $(u,v)$ of $N$ such that $C$ is the set of leaf descendants of v.}.  A network represents a set of clusters $\mathcal{C}$ if it 
represents every cluster in $\mathcal{C}$ (and possibly more). 
The set of {\emph{all}} softwired clusters {represented by a network} can be obtained as follows. 
For a network $N$, we say that a \emph{switching} of $N$ is obtained by, for each reticulation node, deleting all but one of its
incoming edges. Given a network $N$ and a switching $T_N$ of $N$, we say that an edge $(u,v)$ of $N$ represents a cluster $C$ w.r.t. $T_N$ if $(u,v)$ is an edge of $T_N$ and $C$ is the set of leaf 
descendants of $v$ in $T_N$. The set of {all} softwired clusters {represented by} $N$, {denoted $\cC(N)$}, is the set of clusters represented by all edges of $N$ w.r.t. $T_N$, where $T_N$ ranges over all possible switchings \cite{HusonRuppScornavacca10}. Note that the set of all possible switchings of $N$ coincides with the set of all trees displayed by $N$. It is
also natural to define that an edge $(u,v)$ of $N$ represents a cluster $C$ if there exists some switching $T_N$ of $N$ such that $(u,v)$ represents $C$ w.r.t $T_N$.
Note that, in general, an edge of $N$ might represent multiple clusters, and a cluster might be represented by multiple edges of $N$.

\begin{figure*}[t]
  \centering
  \includegraphics[scale=.165]{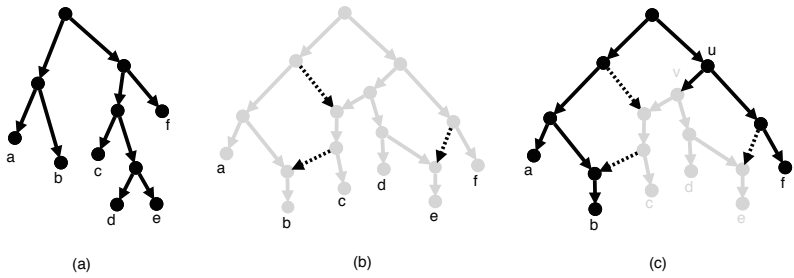}
  \caption{A phylogenetic tree~$T$ (a) and a phylogenetic network~$N$ (b,c); (b) illustrates in grey that~$N$ displays~$T$  (deleted edges are dashed); 
(c)
illustrates that~$N$ represents (amongst others) the cluster $\{c,d,e\}$ in the softwired sense (dashed reticulation edges are ``switched off'').}
  \label{fig:treesclusters}
\end{figure*}

Given a set of clusters $\cC$ on $\cX$, throughout the article  we assume that, for any taxon $x \in \cX$, $\cC$ contains at least one cluster $C$
containing $x$. 
For a set $\mathcal{C}$ of clusters on $\mathcal{X}$ we define $r(\mathcal{C})$ as $\min \{ r(N) | N \text{ represents } \mathcal{C} \}$, we sometimes refer to this as the \emph{reticulation number} 
of $\mathcal{C}$. 
The related concept of {\em level} 
 requires some more background. A directed acyclic graph is \emph{connected} (also called ``weakly connected'') if there is an undirected path 
(ignoring edge orientations) between each pair of nodes. A node (edge) of a directed graph is called a \emph{cut-node} (\emph{cut-edge}) if its removal disconnects the graph. A directed graph 
is \emph{biconnected} if it contains no cut-nodes. A biconnected subgraph~$B$ of a directed graph~$G$ is said to be a \emph{biconnected component} if there is no biconnected subgraph~$B' \neq 
B$ of $G$ that contains~$B$. A phylogenetic network is said to be a 
\emph{${\mbox{level-}\leq k}$ network}  
if each biconnected component has reticulation number less than or equal to $k$.\footnote{Note that to 
determine the reticulation number of a biconnected component, the indegree of each node is computed using only edges belonging to  this biconnected component.}
A 
network is called \emph{simple} if the
removal of a cut-node or a cut-edge creates two or more connected components of which  
at most one is non-trivial (i.e. contains at least one edge). 
A (simple) level-$\leq k$ network $N$ is called a (simple) \emph{ level-$k$ network} if  the maximum reticulation number among the biconnected components of $N$  is precisely~$k$. 
For example, the network in Figure \ref{fig:bcc} is a  level-2 network {(which is not simple)},  the network in Figure 
\ref{fig:gallednetwork} is a simple level-4 network and the network in Figure \ref{fig:level2example} is a simple level-2 network. Note that a tree is a level-0 network. For a set $\mathcal{C}$ of clusters on $\mathcal{X}$ we define $l(\mathcal{C})$, the
\emph{level} of $\mathcal{C}$, as the smallest $k \geq 0$ such that there exists a level-$k$ network that represents $\mathcal{C}$.  It is immediate that for every cluster set $\mathcal{C}$ it holds that 
 $r(\mathcal{C}) \geq l(\mathcal{C})$, because a level-$k$ network always contains at least one biconnected component containing $k$ reticulations. 


We say that two clusters~$C_1,C_2\subset\mathcal{X}$ are \emph{compatible} 
if either~$C_1\cap C_2=\emptyset$ or~$C_1\subseteq C_2$ or~$C_2\subseteq C_1$, {and \emph{incompatible}
otherwise}. Consider a set of clusters~$\mathcal{C}$. 
The \emph{incompatibility graph}~$IG(\mathcal{C})$ 
of~$\mathcal{C}$ is the undirected graph~$(V,E)$ that has node set~$V=\mathcal{C}$
and edge set
$E=\{\{C_1,C_2\}\enspace |\enspace C_1$ \mbox{ and } $C_2$\mbox{ are incompatible\newline} \mbox{clusters in} $\mathcal{C}$\}.
We say that a set of taxa~$\cX' \subseteq\mathcal{X}$ is 
{\emph{compatible} with~$\mathcal{C}$ if every cluster $C\in\mathcal{C}$ is compatible 
with~$\cX'$, and \emph{incompatible} otherwise.}

We say that a set of clusters $\mathcal{C}$ on $\mathcal{X}$ is  {\em separating} if it is incompatible with all sets of taxa $\cX'$ such that $\cX' \subset \cX$ and $|\cX'| \geq 2$.

\begin{figure}[t]
  \centering
  \begin{subfigure}[]
  {
    \centering
    \includegraphics[scale=0.16]{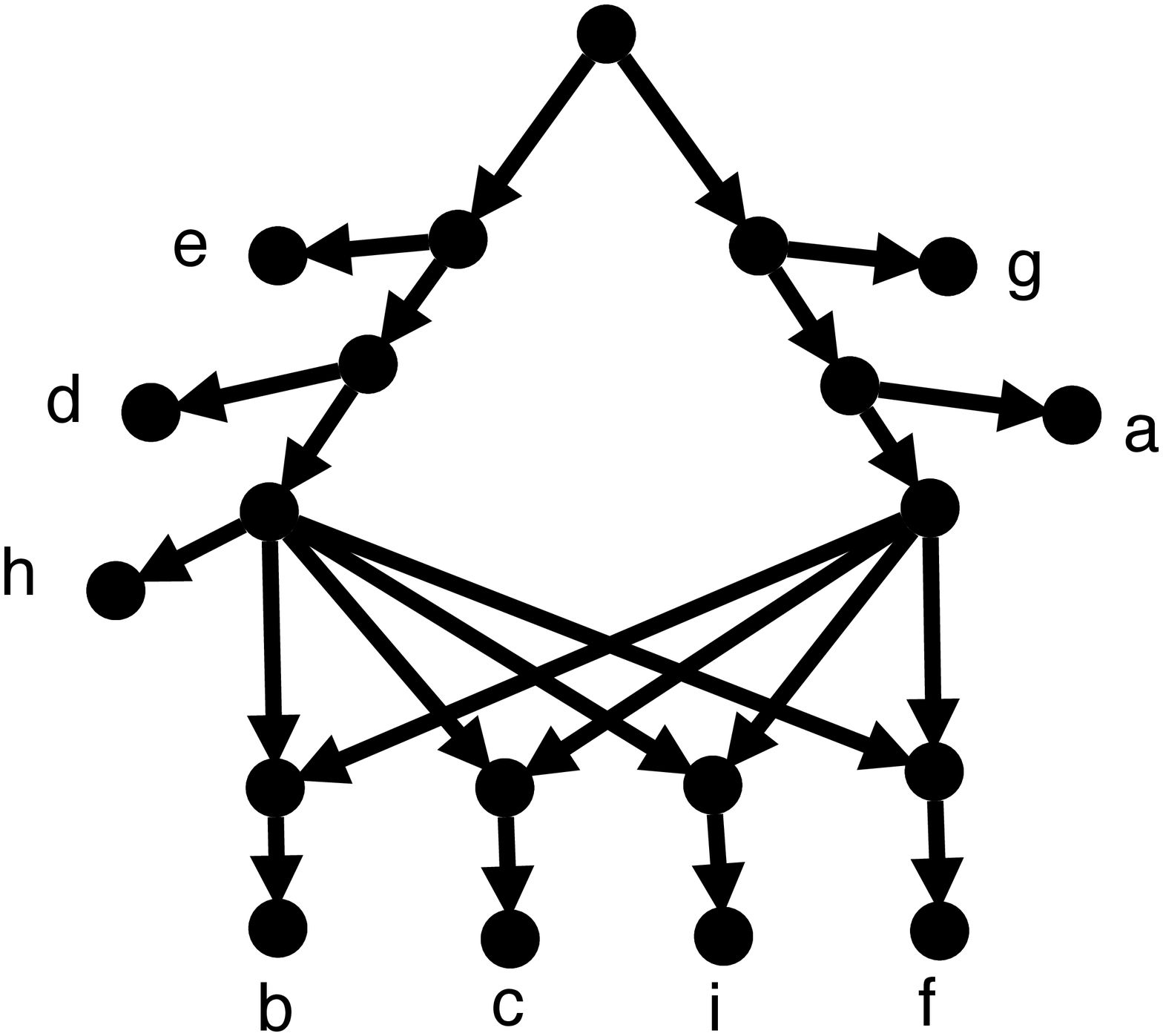}
    \label{fig:gallednetwork}
    \vspace{.5cm}
  }
  \end{subfigure}
  \begin{subfigure}[]
  {
    \centering
    \includegraphics[scale=0.16]{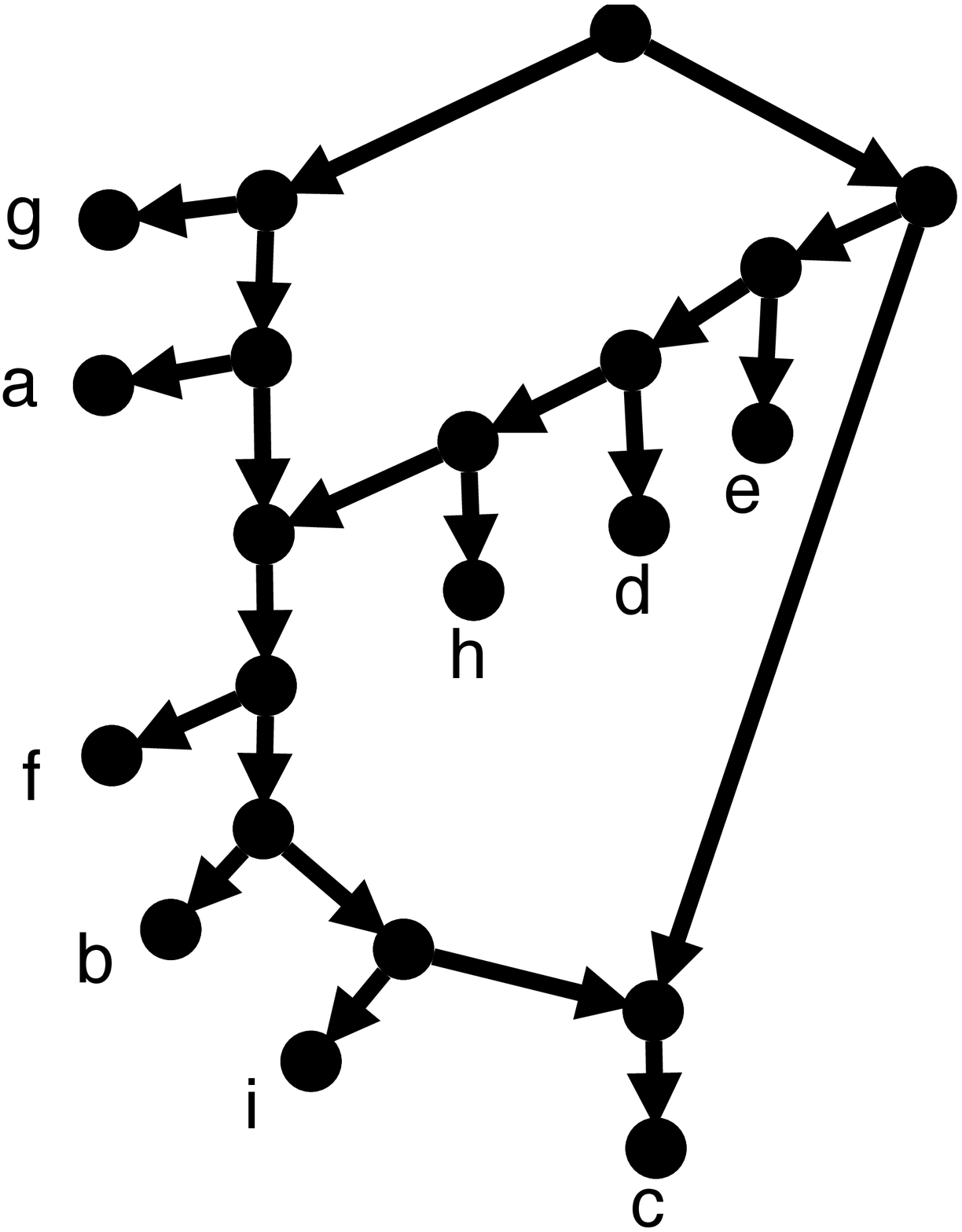}
    \label{fig:level2example}
  }
  \end{subfigure}
  \caption{Two examples of networks that represent, among others,  the set of clusters $\mathcal{C}=\{\{a,b,f,g,i\}$, $\{a,b,c,f,g,i\}$, $\{a,b,f,i\}$, $\{b,c,f,i\}$, $\{c,d,e,h\}$, 
$\{d,e,h\}$, $\{b,c,f,h,i\}$, $\{b,c,d,f,h,i\}$, $\{b,c,i\}$, $\{a,g\}$, $\{b,i\}$, $\{c,i\}$, $\{d,h\}\}$. The network in (a) is a simple level-4 network, and the
network in (b) is a (binary) simple level-2 network.}
  \label{fig:example}
\end{figure}

When we write $f(k)$ we mean ``some function that only depends on $k$''. For simplicity we overload $f(k)$ to refer to multiple different functions with this property. We write $poly(n)$ to mean ``some function $f(n)$ that is polynomial in $n$'', {where $|\cX|=n$}. As in the case of $f(k)$,  we often overload this expression. Indeed, the goal of this article is not to derive exact expressions for the running time, but 
to show that it is bounded above by $f(k) \cdot poly(n)$. It should be noted that the $f(k)$ that we encounter in this article can be \emph{extremely} exponential in $k$. {Also, $|\cC|$ can in general
be exponentially large as a function of $n$, but (as we shall see in due course) it is reasonable to assume that $|\cC|$ is bounded above by $f(k) \cdot poly(n)$ when the parameter $k$ (reticulation number or level) is fixed}. 

The next lemma ensures that, if our goal is to find a network representing a set of clusters and minimizing the level or the reticulation number, we can restrict our attention to binary networks: 

\begin{lemma}
\label{lem:binaryOK}
\cite{elusiveness} Let $N$ be a  network on $\cX$. Then we can
transform N into a binary  network $N'$ such that $N'$ has the same reticulation number
and level as $N$ and all clusters represented by $N$ are also represented by $N'$.
\end{lemma}

Thanks to Lemma \ref{lem:binaryOK}, we may assume that there exists a {binary} network $N$ with reticulation number  $r(\cC)$ (or  with  level $l(\cC)$ if we are interested in level minimization) that represents $\cC$. We henceforth restrict
our analysis to binary networks and, {except in places where it might cause confusion to not be explicit}, we will not emphasize again that we only deal with this kind of network.  



\section{Minimizing level is fixed parameter tractable}
\label{sec:levelFPT}

The aim of this section is to show that level-minimization is  fixed parameter tractable.  
To compute $l(\cC)$, we will repeatedly query, \emph{``Is $l(\cC)=k$? If so, construct a network with level equal to $k$ that represents $\cC$''} for $k = 0,\ldots,l(\cC)$, where $k$ starts at 0 and is incremented by 1 until the query is answered positively. 
Assuming that the queries are correctly answered, this process will terminate after $l(\cC)+1$ iterations. Hence, to prove an overall running time of $f(l(\cC)).poly(n)$, it is sufficient to show that
for each $k$ we can correctly answer the query in time at most $f(k) \cdot poly(n)$. {Note that $r(\cC)=l(\cC)=0$ if and only if all the clusters in $\cC$ are pairwise mutually compatible, which can be easily
checked in time $poly(n)$, so we henceforth assume that $k \geq 1$}.

{The high-level idea is the following. In \cite{cass,HusonRuppScornavacca10} it is shown that level-$k$ networks can be constructed using a divide and conquer strategy. Informally, the idea is
to construct a level-$\leq k$ network for each connected component of $IG(\mathcal{C})$ and then to combine these into a single network. The clusters in each connected component first have to be processed,
which creates (for each component) a separating set of clusters. From Lemma 1 of \cite{elusiveness}, we know that, if a level-$k$ network representing a 
separating set of clusters $\cC$ on $\cX$ exists, a simple level-$k$ network representing $\cC$ has to exist. This network will never have two or more taxa with the same parent \cite{elusiveness}.
The transformation underpinning Lemma \ref{lem:binaryOK} furthermore allows us to assume that this simple level-$k$ network is binary.
Hence, the divide and conquer strategy essentially reduces to constructing binary simple level-$\leq k$ networks for separating sets of clusters (and then combining them into a single network).}

{In Section \ref{subsec:simpleFPT} we show how to construct a simple level-$k$ network in time $f(k) \cdot poly(n)$ from a separating set of clusters. Subsequently
we show in Section \ref{subsec:generalFPT} how to combine these networks in time $f(k) \cdot poly(n)$ into a single level-$k$ network.}

\subsection{\emph{Constructing simple networks from separating cluster sets}}
\label{subsec:simpleFPT}

Before proving the main result of this paper, we need to prove some preliminary results.

\begin{proposition}
\label{prop:checkReprCluster}
Given a simple level-$k$ network $N$ and a set of clusters $\cC$ on $\cX$, checking whether  $\cC$ is represented by $N$ can be done in  time $f(k) \cdot poly(n)$, where $n=|\cX|$.
\end{proposition}
\begin{proof}
Note that there are at most $2^{k}$ trees displayed by $N$ and each tree represents at most $2(n-1)$ clusters. This means that $|\cC(N)|$ is at most $2^{k+1}(n-1)$. Since $N$ cannot represent $\cC$ if $|\cC|>|\cC(N)|$, checking whether $\cC \subseteq \cC(N)$ takes at most $f(k) \cdot poly(n)$ time.
\end{proof}

Thus, if $|\cC| > 2^{k+1}(n-1)$, since $\cC$ is assumed to be separating,  it is not possible that $l(\cC) = k$  and we can immediately answer ``no'' to the query. We thus
henceforth assume that $|\cC| \leq 2^{k+1}(n-1)$ i.e. that $\cC$ contains at most $f(k) \cdot poly(n)$ clusters.\\
\\
If all the leaves of a binary simple level-$k$ network $N$ are removed and all nodes with both indegree and outdegree equal to 1 are deleted, the resulting structure is called a \emph{level-$k$ generator} as defined in \cite{lev2TCBB}. See Figure \ref{fig:gen} for the level-1 and level-2 generators. The number of level-$k$ generators is bounded by $f(k)$ \cite{Gambette2009structure}\footnote{Note that the number of level-$k$ generators grows
rapidly in $k$, lying between $2^{k-1}$ and $k!^2 50^k$ \cite{Gambette2009structure}.}. 

\begin{figure}[t]
  \centering
  \includegraphics[scale=.2]{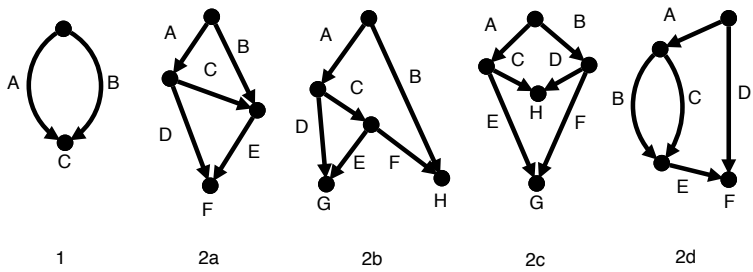}
  \caption{The single level-1 generator and the four level-2 generators. Here the sides have been labelled with capital letters.}
  \label{fig:gen}
\end{figure}
 
 The \emph{sides} of a level-$k$ generator are defined as the union of its edges (the \emph{edge} sides) and its nodes of indegree-2 and outdegree-0 (the \emph{node sides}).  The number of sides in a generator is bounded by $f(k)$, because the sum of its vertices and edges is linear in $k$ \cite{reflections}. 

\begin{definition}
\label{def:Nk}
The set $\mathcal{N}^{k}$ (for $k \geq 1$) is defined as the set of all  networks that can be constructed by choosing some level-$k$
generator $G$ and then applying the following \emph{leaf hanging} transformation to $G$ such that
each taxon of $\cX$ appears exactly once in the resulting network. (This is essentially identical
to the definition given in \cite{reflections}, which is only a superficial refinement of the definition given in \cite{lev2TCBB}). 
\begin{enumerate}

\item First, for each pair $u, v$ of vertices in $G$ connected by a single edge $(u, v)$,
replace $(u, v)$ by a path with $l \geq 0$ internal vertices and, for each 
such internal vertex $w$, add a new leaf $w'$, an edge $(w, w')$, and
label $w'$ with some taxon from $\cX$.  All the taxa added in this way
are ``on side $s$'' where $s$ is the side corresponding to the edge $(u,v)$. (It is also permitted that the path has zero
internal nodes i.e. that the side remains empty).
\item Second, for each pair $u, v$ of vertices in $G$ connected by two edges, treat the two edges as in step 1, but ensure that at least one of the two paths does not have zero 
internal nodes. 
\item Third, for each vertex $v$ of $G$ with indegree 2 and outdegree 0 add a new
leaf $y$, an edge $(v, y)$ and label $y$ with a taxon $x \in \cX$; we say ``taxon $x$ is
on side $s$'' where $s$ is the side corresponding to vertex $v$.
\end{enumerate}
\end{definition}

{ The main
reason for step 2 in Definition \ref{def:Nk} is to ensure that multi-edges in generators do not survive in the final network,
because our definition of phylogenetic network does not allow multi-edges.}
The following lemma follows directly from the results in Section 3.1 of \cite{lev2TCBB}:

\begin{lemma}
\label{lemma:NkSet}
The set $\mathcal{N}^{k}$ (for $k \geq 1$) is equal to the set of all binary simple level-$k$ networks.
\end{lemma}

For example, the simple network in Figure \ref{fig:level2example} has been obtained from generator $2a$ (see Figure \ref{fig:gen}) by putting 0 taxa on sides $A$ and $D$, 1 taxon on side $F$, 2 taxa on side $B$ and 3 taxa on sides $C$ and $E$.

By Lemma 1 of \cite{elusiveness} and Lemmas 1 and 2, we have the following:
\begin{corollary}
\label{cor:genIsSufficient}
Let $\cC$ be a 
separating set of clusters on $\cX$, such that $l(\cC) \geq 1$. Then there exists a network $N$ in $\mathcal{N}^{l(\mathcal{C})}$ such that
$N$ represents $\cC$.
\end{corollary}

Given a taxon set $\cX$, we call any network resulting from adding all taxa in $\cX$ to sides of a generator $G$ (in the sense of Definition \ref{def:Nk}) a {\em completion} of $G$ on $\cX$. Here we call a side that receives $\geq 2$ taxa a \emph{long side}, a side that receives 1 taxon a \emph{short side} and a side that receives 0 taxa an \emph{empty side}. Figure \ref{fig:level2example} is thus a completion of generator $2a$,
where sides $A$ and $D$ are empty, side $F$ is short, and sides $B, C, E$ are long. Note that node sides (such as $F$ in the example) are always short, but not all short sides are node sides i.e.
edge sides can be short too.
 
Given a generator $G$, we call  \emph{a set of side guesses} for $G$, denoted by $S_G$, a set of guesses about the type of each side of $G$ (i.e. whether it is empty, short, or long). 
A completion $N$ of $G$ on $\cX$ \emph{respects  $S_G$} if all {sides} that are long in  $S_G$ receive at least 2 taxa in $N$, {sides} that are short in  $S_G$ one taxon and {empty sides} zero taxa. 
Then we have the following result:
\begin{observation}
\label{obs:searchingSpace}
Searching in the space of all binary simple level-$k$ networks on $\cX$ is equivalent to searching in the space of  all completions of a level-$k$ generator $G$ respecting a set of side guesses $S_G$, iterating overall  all sets of side guesses for a generator and  all level-$k$ generators.
\end{observation}




Let $G$ be a level-$k$ generator and let $S_G$ be a  set of side guesses for $G$. We say that the pair $(G,S_G)$ is \emph{side-minimal} w.r.t. a separating cluster set $\cC$ on $\cX$ and  $k$, if there exists a completion $N$ of $G$ on $\cX$ respecting $S_G$ that is a  level-$k$ network representing $\cC$ and, amongst all simple level-$k$ networks that represent $\cC$, $N$ has a minimum number of long sides, and (to further break ties) amongst those networks it has a minimum number of short sides.

We define an \emph{incomplete network} as 
a generator $G$, a set of side guesses $S_G$, a set of \emph{finished sides} (i.e. those sides for which we have already decided that no more taxa will placed on them), a set
of \emph{future sides} (i.e. those short and long sides that have had no taxa allocated yet) and at most one long side on which at least one taxon has already been placed but where we might still want to add some more taxa. We call this the \emph{active side}. A \emph{valid completion} of an incomplete network is an assignment of the unallocated taxa to the future sides and (possibly) \emph{above} the taxa already placed on the active side, that respects $S_G$ 
and such that the
resulting network (which we call the \emph{result} of the valid completion) represents $\cC$. Informally, the
result of a valid completion is any network on $\cX$  respecting $S_G$ and  representing $\cC$   that is obtained by respecting all
placements of taxa made thus far. 

For example, consider again the network in Figure \ref{fig:level2example}. Let $N$ be the network in that figure and let $N'$ be the
network obtained from $N$ by deleting taxa $c, d, e$ and suppressing the resulting vertices with indegree and outdegree both equal to 1. Let $G$ be generator 2a,
and let $S_G$ be the set of side-guesses where sides $A$ and $D$ are empty, side $F$ is short, and sides $B, C, E$ are long. Sides
$A,B,D,E$ are finished, $F$ is a future side and $C$ is the active side. In particular, we can perform a valid completion of this incomplete network $N'$ by putting
taxa $d$ and $e$ above taxon $h$ on side $C$ and then putting $c$ on side $F$. In this case, $N$ is the result of the completion, although in general an incomplete network
might have many valid completions, or none.\\
\\
Given a cluster set $\cC$, we write $x \rightarrow_{\cC}  y$ if and only if every non-singleton cluster in $\cC$ containing $x$, also contains $y$. Then we have the following result.
\begin{proposition}
\label{prop:noCyclesRightarrow}
Given a separating  set of clusters  $\cC$ on $\cX$ and an ordered set of distinct taxa of $\cX$ $(x_1, \dots, x_k)$ such that  $k\geq 2$ and $x_i \rightarrow_{\cC}  x_{i+1}$ for $ 1 \leq i \leq (k-1)$. Then $x_k \not\rightarrow_{\cC}  x_1$.
\end{proposition}
\begin{proof}
If $x_i \rightarrow_{\cC}  x_{i+1}$ for $ 1 \leq i \leq (k-1)$ and $x_k \rightarrow_{\cC}  x_1$, this means that the set $\cX'= \cup_{i=1}^{k}  x_i$ is compatible with $\cC$. Since $|\cX'|\geq 2$, we have  a contradiction. 
\end{proof}

The following observations will be useful to prove Lemma \ref{lemma:subroutineIsCorrect}.

 
\begin{observation}
\label{obs:sameSide}
Let $N$ be  a  phylogenetic network on $\cX$ representing a  set of clusters $\cC$ on $\cX$ constructed 
by choosing some level-$k$
generator $G$ and then applying the  \emph{leaf hanging} transformation described in Definition \ref{def:Nk} to $G$. 
If two taxa $x$ and $y$ in $\cX$ are on the same side of the generator underlying $N$ and the parent of $x$ is a descendant of the parent of $y$, then $y \rightarrow_{\cC}  x$. 
\end{observation}

Given a simple phylogenetic network $N$, we say  that a side $s'$ is \emph{reachable} from a side $s$ in $N$ if  there is a directed  path in the generator underlying $N$ from the head of side $s$ to the tail of side $s'$.  

\begin{observation} 
\label{obs:when_l'_impl_l}
Let $N$ be  a  phylogenetic network on $\cX$ representing a  set of clusters $\cC$ on $\cX$ constructed 
by choosing some level-$k$
generator $G$ and then applying the  \emph{leaf hanging} transformation described in Definition \ref{def:Nk} to $G$. Moreover,  let $x$ and $y$ be two taxa of $\cX$ on the same side $s$ of the generator underlying $N$ such that  $y \rightarrow_{\cC}  x$ and let $z$ be a taxon on a side $s' \neq s$ such that $s'$ is not reachable from $s$ and $z \rightarrow_{\cC}  x$. Then we have that  $z \rightarrow_{\cC}  y$. 
\end{observation} 
\begin{proof}
Since $z \rightarrow_{\cC}  x$, we know that every non-singleton cluster that contains $z$ also contains $x$. Now, let $C$ such a cluster. $C$ is
represented by some tree $T$ displayed by $N$, so some edge $e$ in $T$ is such
that $C$ is the set of all taxa reachable from directed paths from the head of $e$. Now,
$z$ and $x$ are both in $C$, so there a directed path from the head of $e$ to $z$ and a directed path from
the head of $e$ to $x$. Since $s'$ is not reachable from $s$, the only way that such a directed path can reach $x$ is via the parent of $y$, hence the fact
that $z \rightarrow_{\cC}  y$.
\cs{If no cluster $C$ containing $z$ and $x$ exists, since $z \rightarrow_{\cC}  x$ we have that the only cluster containing $z$ is the singleton cluster $\{z\}$. Then, obviously, $z \rightarrow_{\cC}  y$ too.}
\end{proof}


\begin{observation}
\label{obs:onlyCluster}
Let $N$ be  a  phylogenetic network on $\cX$ representing a  set of clusters $\cC$ on $\cX$ constructed 
by choosing some level-$k$
generator $G$ and then applying the  \emph{leaf hanging} transformation described in Definition \ref{def:Nk} to $G$. 
Let $x$ and $y$ be  two taxa in $\cX$  on the same side $s$ of the generator underlying $N$ such that there exists an edge $e$ from the parent of $y$ to the parent of $x$. Then $e$ represents all   clusters in $\cC$ containing $x$ but not $y$.
\end{observation}

\begin{observation}
\label{obs:size2Incomp}
Let $\cC$ be a separating cluster set on $\cX$. Then every size-2 subset of  $\cX$ is incompatible with $\cC$. 
\end{observation}

\begin{figure}[t]
  \centering
  \begin{subfigure}[]
  {
    \centering
    \includegraphics[scale=0.16]{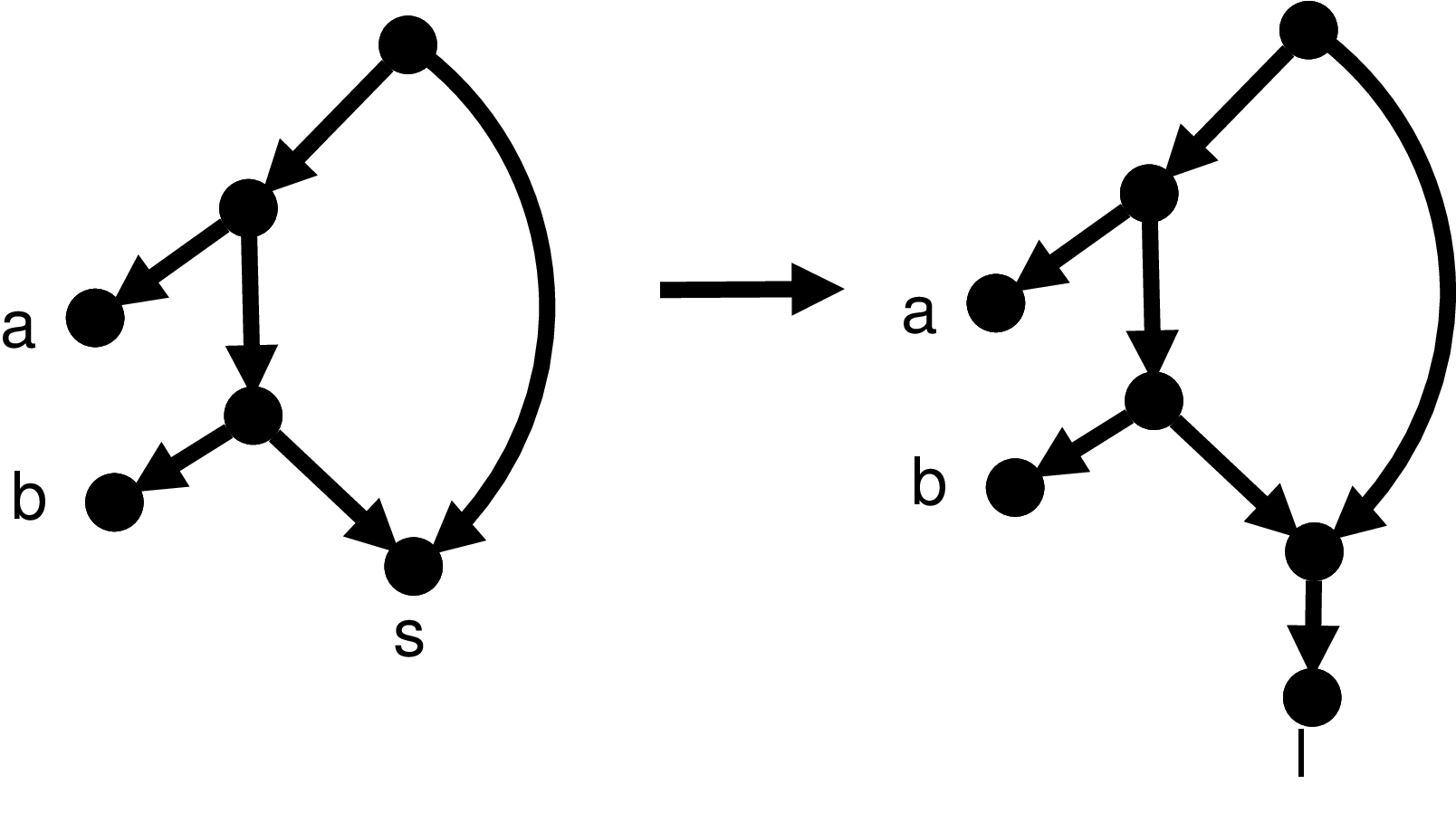}
    \label{fig:nlsFirst}
    \vspace{.5cm}
  }
  \end{subfigure}
  \begin{subfigure}[]
  {
    \centering
    \includegraphics[scale=0.16]{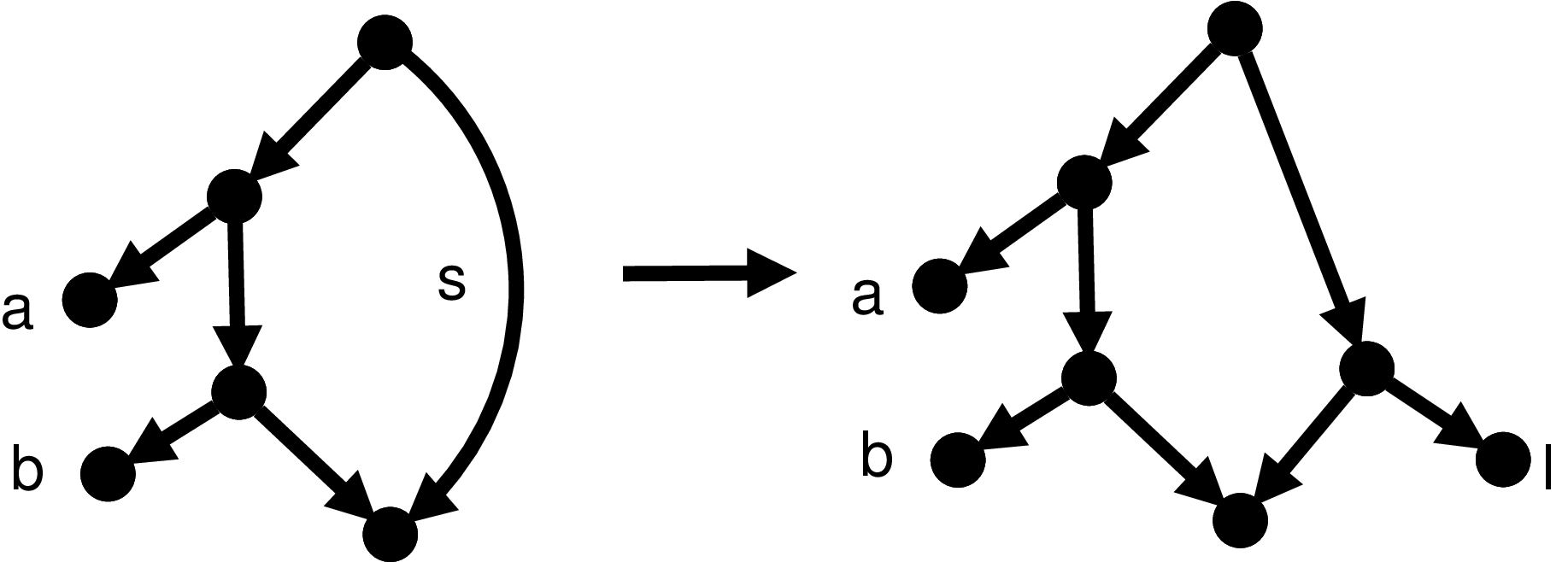}
    \label{fig:nlsSecond}
  }
  \end{subfigure}

  \begin{subfigure}[]
  {
    \centering
    \includegraphics[scale=0.16]{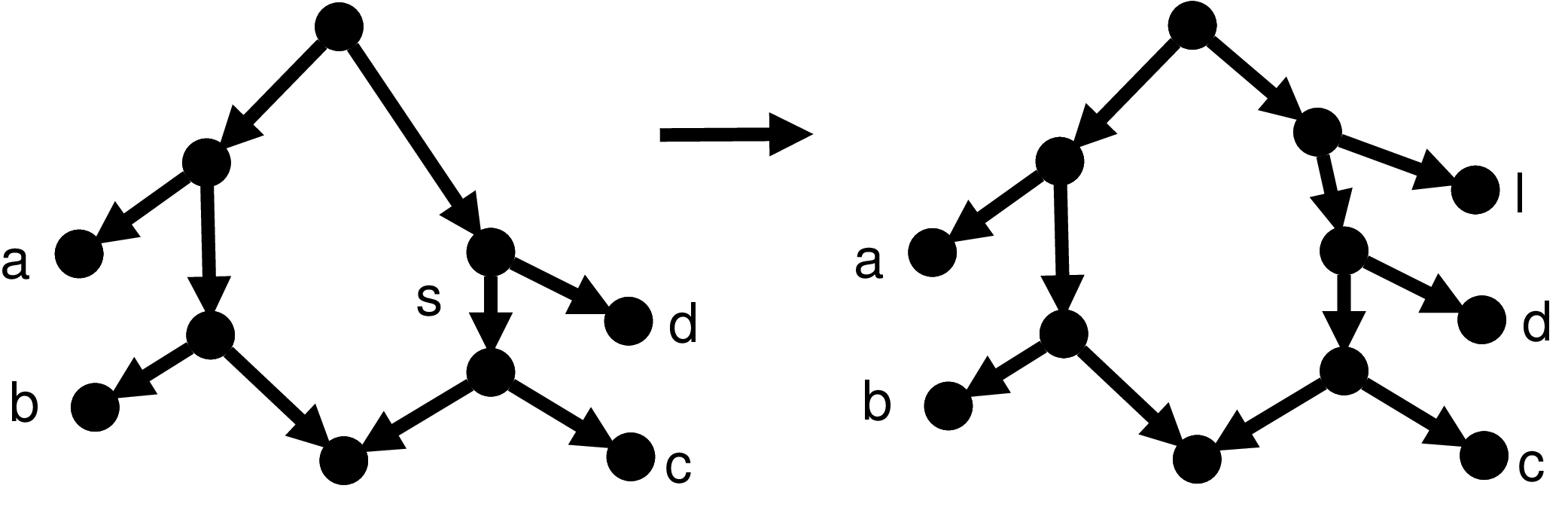}
    \label{fig:nlsThird}
  }
  \end{subfigure}

  \caption{Three examples of the $N(l,s)$ operation. (a) $N(l,s)$ when $s$ is an unfinished short node side; (b) $N(l,s)$ when $s$ is an unfinished short edge side (or a long side that does not
yet have any taxa); (c) $N(l,s)$ when $s$ is a long side that already has at least one taxon.}
  \label{fig:nlsExamples}
\end{figure}

Let $N$ be a simple phylogenetic network $N$, $l$ a taxon and $s$ a side of the generator $G$ underlying $N$. We denote by $N(l,s)$ the following operation, where we exclude the case
from consideration where $s$ is a short side that already has a taxon on it. If $s$ is a short side, then $N(l,s)$ is simply the network obtained by putting $l$ on side $s$ (in the sense of Definition \ref{def:Nk}).
Otherwise, $s$ is a long side, and then $N(l,s)$ is the network obtained by placing $l$ ``just above'' the highest taxon on side $s$. If there are not yet any taxa on side $s$ then
we simply let $l$ be the first taxon on side $s$. (See Figure \ref{fig:nlsExamples} for clarification).\\
\\
We are now ready to analyse Algorithm \ref{algo:addOnSide}, which is a critical subroutine. Let us assume that we have an incomplete network $N$ with an active side $s$ (which is by definition
long) such that \underline{all} long sides $s' \neq s$ that are reachable from $s$, are \underline{finished}. These preconditions will be motivated in due course. Informally, Algorithm \ref{algo:addOnSide} lets us decide whether
we should continue adding taxa to the top of the active side, or stop and declare it finished. (In fact, the algorithm is rather more complicated than that, because a side-effect of the algorithm
is that it sometimes adds taxa to unfinished short sides, irrespective of whether it has chosen to add a taxon to the top of the active side). Algorithm \ref{algo:completeSide} will
repeatedly call Algorithm \ref{algo:addOnSide} until it finally declares the active side finished. This is all ultimately driven by the main algorithm, Algorithm \ref{algo:core}, which - amongst
other tasks - then identifies and initialises (i.e. places a first taxon at the bottom of) a new active side.

\begin{lemma}
\label{lemma:subroutineIsCorrect}
Let $\cC$ be a separating set of clusters on $\cX$ and let $k$ be the first integer for which a level-$k$ network representing $\cC$ exists. 
Let $N$ be an incomplete network such that 
its underlying generator $G$ and set of side guesses $S_G$ are such that $(G,S_G)$ is side-minimal w.r.t. $\cC$ and $k$, and let $s$ be an active side of $N$. 
Then, if a valid completion for $N$ exists, Algorithm \ref{algo:addOnSide}  computes a set of (incomplete) networks $\mathcal{N}$ such that  
 this set contains at least one  network  for which a valid completion exists. 
\end{lemma}

\begin{algorithm}[] 
$\cX' \gets{} \cX \setminus \cX(N)$\;
 $x_i \gets{}$ the most recent taxon inserted on side $s$\;
 $L \gets{} \ \{ l \in \cX' |$ $l \rightarrow_{\cC}  x_{i}\}$\;
$L' \gets{} \{ l \in L | \text{ there does not exist } l' \in L \text{ such that } l' \neq l \text{ and } l \rightarrow_{\cC}  l' \}$\;
$U \gets \{s'| s'\neq s$ is a side of $N$ that is not yet finished and is reachable from $s\}$\;
\ForEach{ $l \in L'$} {$S(l) =  \bigcup \{ C \in \cC|$ $x_i \in C \text{ and } l \not \in C\}$\;  $B(l) = \cX' \cap S(l)$.}

\If{$U = \emptyset$} {
\lIf{$|L'| \neq 1$}{declare $s$ as finished in $N$ and \Return $N$\;}
$l \gets{}$  removeFirst($L'$)\;
\lIf{$B(l) \neq \emptyset$}{declare $s$ as finished in $N$ and \Return $N$\;}
\lIf{$N(l,s)$ does not represent $\cC$ restricted to $\cX(N) \cup \{l\}$}
{declare $s$ as finished in $N$ and \Return $N$\;}
\Else{
\Return $N(l,s)$\;}
}
\Else{
\If{$L' = \emptyset$}{declare $s$ as finished in $N$ and \Return $N$\;}
\If{$|L'|\geq 2$}{
	$\mathcal{N} \gets{} N$, where $s$ is declared as finished\;
	\If{$|L'| \leq |U|$}{
		$\mathcal{N'} \gets{}$ the set of networks obtainable from $N$ by allocating all taxa in $L'$ to sides in $U$\;
		$\mathcal{N} \gets{} \mathcal{N} \cup \mathcal{N'}$\;
		}
	\If{$|L'|-1 \leq |U|$}{
		\ForEach{$l \in L'$}{
			$\mathcal{N'} \gets{}$ the set of networks obtainable from $N(l,s)$ by allocating all taxa in $L'\setminus \{l\}$ to sides in $U$\;
		$\mathcal{N} \gets{} \mathcal{N} \cup \mathcal{N'}$\;
	}
	}
	\Return $\mathcal{N}$\;
	}
\If{$|L'|=1$}{
$l \gets{}$  removeFirst($L'$)\;
$\mathcal{N} \gets{} \emptyset$\; 
\If{$B(l) \neq \emptyset$}{ 
	$\mathcal{N} \gets{} N$, where $s$ is declared as finished\;
	 \ForEach{side $s' \in U$}{
		$\mathcal{N} \gets{} \mathcal{N} \cup N(l,s')$\; 
}
	\If{$|B(l)| \leq |U|$}{
		$\mathcal{N'} \gets{}$ the set of all networks obtainable from $N(l,s)$ by allocating all taxa in $B(l)$ to sides in $U$\;
		$\mathcal{N} \gets{} \mathcal{N} \cup \mathcal{N'}$\;
	}
}
\Else{
\If{$N(l,s)$ does \emph{not} represent $\cC$ restricted to $\cX(N) \cup \{l\}$}{ 
	$\mathcal{N} \gets{} N$, where $s$ is declared $s$ as finished\;
	 \ForEach{side $s' \in U$}{
		$\mathcal{N} \gets{} \mathcal{N} \cup N(l,s')$\; 
}
}
\Else{
$D \gets{}$  an arbitrary set of $|U|$  taxa such that $D \cap \cX = \emptyset$\; 
$N^*(l,s) \gets{}$ a network obtained from $N(l,s)$ by arbitrarily and bijectively assigning each taxon in $D$ to a side in $U$\;
 $\cC^{*} \gets{}\{ C\in  \cC$ such that $x_i \in C$, $l \not \in C$, and $C \subseteq \cX(N)\}$\;
 \If{ $N^*(l,s)$  does not
represent $\cC^{*}$}{ 
	$\mathcal{N} \gets{} N$, where $s$ is declared $s$ as finished\;
	 \ForEach{side $s' \in U$}{
		$\mathcal{N} \gets{} \mathcal{N} \cup N(l,s')$\; 
	}
	}
\Else{
$\mathcal{N \gets}$ $N(l,s)$\;} 
}

}
\Return $\mathcal{N}$\;
}
}
\caption{
addOnSide($N,s$)\label{algo:addOnSide}
}
\end{algorithm}
\begin{proof}

Recall that, from Corollary \ref{cor:genIsSufficient}, we can restrict our search to networks in $\mathcal{N}^k$. 
We write $\mathcal{X}(N)$ to denote the set of taxa present in a (incomplete) network $N$. For a set of clusters $\cC$ on $\cX$ and a subset $\cX' \subseteq \cX$, we define the
\emph{restriction} of $\cC$ to $\cX'$ as $\{ C \cap \cX' | C \in \cC\}$. We start the proof by analyzing  the case when $U = \emptyset$ (see Algorithm \ref{algo:addOnSide} for the definition of $\cX'$, $U$, $B(l)$,  etc).

\paragraph{\textbf{Case $\boldsymbol{U = \emptyset}.$}} 
Suppose $|L'| \neq 1$. If $\boldsymbol{|L'| = 0}$ then there are two possibilities. 
If $L= \emptyset$ then clearly no taxon $l$ can be placed  directly above
$x_{i}$ on $s$, because that would mean $l \rightarrow_{\cC}  x_i$, and thus $l \in L$, contradiction. Hence the only correct move is to declare that the side $s$ is finished and return $N$. If $L \neq \emptyset$ then, since $|L'| = 0$, we have that, for every $l \in L$ there exists
some $l' \in L$ such that $l \neq l'$ and $l \rightarrow_{\cC}  l'$. Clearly the $\rightarrow_{\cC} $ relation is not allowed to create cycles in $L$, because otherwise the set of taxa in the cycle would form a cluster compatible with $\cC$ (see Proposition \ref{prop:noCyclesRightarrow}). Suppose we start at an arbitrary taxon in $L$ and perform a non-repeating walk on the taxa of $L$ by following the $\rightarrow_{\cC} $ relation. Given that $L$ is of finite size and this walk cannot visit a taxon of $L$ that it has already visited earlier in the walk (thus creating a cycle), we will find a taxon $l \in L$ such that there is no $l' \in L$ such that $l \neq l'$ and $l \rightarrow_{\cC}  l'$,
meaning that $l \in L'$, contradiction. So the case that $L \neq \emptyset$ but $L' = \emptyset$, cannot actually happen. 
Now, consider the case that $\boldsymbol{|L'| \geq 2}$. Algorithm \ref{algo:addOnSide} will always end the side $s$ and return $N$ in this case. Indeed, no valid completion of $N$ can have some taxon $p$ that has not yet been allocated above $x_i$ on side $s$. 
Suppose this is not true. 
Clearly, from Observation \ref{obs:sameSide}, $p \rightarrow_{\cC}  x_i$, so $p \in L$. In this case, all taxa in $L'$ are either equal to $p$, or underneath $p$ and above $x_i$. Indeed, let $l \neq p$ be a taxon in $L'$ and suppose,  for the sake of contradiction, that $l$ is above $p$ on side $s$ or on another side $s'$. 
If $l$ is above $p$ on side $s$, then from Observation \ref{obs:sameSide} we have that $l \rightarrow_{\cC}  p$. 
If $l$ is  on another side $s'$, the fact that $|U|=0$ implies that there is no room under side $s$ so, by Observation \ref{obs:when_l'_impl_l} we have that $l \rightarrow_{\cC}  p$.  
Thus, in both cases (i.e. if $l$ is above $p$ on side $s$ or on a different side $s'$) we have that $l \rightarrow_{\cC}  p$, meaning that $l \not \in L'$, contradiction. 
We can hence conclude that each taxon in $L'$ is either equal to $p$, or underneath $p$ and  above $x_i$ in any completion of $N$ where $p$ is on $s$. 
But, however one arranges two or more taxa on one side, 
at least one taxon will imply another taxon in the sense of the $\rightarrow_{\cC} $ relation. 
More formally,  in any case there exist two taxa  $l$ and $l'$ in $L'$ such that $l \neq l'$ and $l \rightarrow_{\cC}  l'$. This implies that $l \not \in L'$, contradiction. This concludes the correctness of the case $|L'| \neq 1$.

We now consider the  case when $\boldsymbol{|L'| = 1}$. Let $l$ be the only taxon in $L'$. 
In this case, Algorithm \ref{algo:addOnSide} will return $N$ if $\boldsymbol{B(l) \neq \emptyset}$. Indeed, no valid completion of $N$ exists where 
 one or more taxa are  placed above $x_i$ on $s$. 
Suppose this is wrong. In that case, observe that in every valid completion $l$ always has to be the taxon directly above $x_i$. Indeed, if there was some valid completion such that $l$ is not directly above $x_i$, then 
there would exist some taxon $l' \neq l$ such that  $l' \rightarrow_{\cC}  x_i$ (from Observation \ref{obs:sameSide}) and $l \rightarrow_{\cC} l'$ (as before, this follows   from  the fact that  $U=\emptyset$ and from Observations   \ref{obs:sameSide} and \ref{obs:when_l'_impl_l}).   
This would mean that $l \not \in L'$, contradiction.  So we assume that $l$ is directly above $x_i$. Now,
since $B(l) \neq \emptyset$, then there is some cluster in the input that contains $x_i$, does not contain $l$, and contains some not-yet allocated taxon distinct from $l$. 
From Observation \ref{obs:onlyCluster}, the only edge that can represent such a cluster is the edge $e$ 
 between the parents of $x_i$ and $l$. But all the clusters  represented by $e$ consist only of already-allocated taxa, because $U = \emptyset$. This means that adding $l$ on side $s$ will only lead us to construct non-valid completions. Hence we
conclude that, if $B(l) \neq \emptyset$, all valid completions of $N$ do not contain any other taxon on $s$ and ending the side $s$ is the right choice.  

Now consider the case $\boldsymbol{B(l) = \emptyset}$ and let $\cC'$ be $\cC$ restricted to $\cX(N) \cup \{l\}$. If $\boldsymbol{N(l,s)}$  \textbf{does not represent $\boldsymbol{\cC'}$} we are definitely correct to declare the side $s$ as finished and return $N$.
Indeed, 
all valid completions of N do not contain any other taxon on $s$. 
 Suppose it is not correct. Then there exists a valid completion of  $N$ where at least one taxon is above $x_i$ on $s$. Again, for the same reasons as above we assume that $l$ is always the taxon directly
above $x_{i}$. Since $N(l,s)$ does not represent $\cC'$, this incompatibility cannot be eliminated by adding more taxa, hence we conclude that there are no valid completions of $N$ with taxa above $x_i$ on side $s$. Hence, ending the side $s$ is 
the only correct option.  
Suppose now that $\boldsymbol{N(l,s)}$ \textbf{does represent} $\boldsymbol{\cC}'$; Algorithm \ref{algo:addOnSide} adds $l$ above $x_i$ on side $s$, 
and does \emph{not} declare $s$ as finished. This conclusion can only be incorrect if \emph{all} valid completions
require that $l$ is \emph{not} directly above $x_i$. We observe that in any valid completion of $N$  there can be no taxon $l'\neq l$ directly above $x_i$ on $s$, because otherwise, as before, since $U = \emptyset$ we will have that $l \rightarrow_{\cC}  l' \rightarrow_{\cC}  x_i$
and hence $l \not \in L'$, contradiction. So all valid completions terminate the side at $x_i$. Let $N'$ be an arbitrary valid completion of $N$ and denote by $N''$ the network obtained from $N'$ by moving $l$, wherever it is, just
above $x_i$. Firstly, we claim that $N''$ still represents $\cC$. Recall that $l \rightarrow_{\cC}  x_i$, so the only potential problem  is with clusters in $\cC$ that contain $x_i$ but do not
contain $l$.  Let $C$ be such a cluster not represented by $N''$. Suppose $C \not \subseteq \cX(N) \cup \{l\}$. But in this case we would have $B(l) \neq \emptyset$ in $N$, contradiction. 
So the only possibility is that $C \subseteq \cX(N) \cup \{l\}$. Clearly $C$ was in $\cC'$ and was
thus represented by  $N(l,s)$. Moreover, from Observation \ref{obs:onlyCluster}, the edge  that represents $C$ in $N(l,s)$ is the edge between the parents
of $l$ and $x_i$.  Given that $U = \emptyset$, no more taxa can
be added ``underneath'' side $s$ and  this edge still represents $C$ in $N''$ because $N'$ is a  valid completion of $N$. Hence moving $l$ in the way \emph{is} safe in terms of cluster representation. 

Secondly, we claim that moving $l$ in this way does not alter the side types i.e. the empty/short/long sides before moving $l$ remain empty/short/long after moving $l$. To
see this, note that moving $l$ from its original location reduces the number of taxa by 1 on some side, and increases the number of taxa of $s$ by 1. Side $s$ is by assumption already long, so remains so. The side of $N'$ containing $l$ cannot 
change from being long to being short in $N''$, because this lowers the total number of long sides, and by assumption the pair $(G,S_G)$ underlying $N$ is side-minimal. Similarly it cannot change from being short to being empty, because this leaves the
number of long sides the same but reduces the number of short sides, again contradicting the assumption 
that $(G,S_G)$ is side-minimal. Combining these two claims - that moving $l$ is safe for cluster representation and does not alter the side types nor the underlying generator - let us conclude that there \emph{is} a valid completion for $N$ in which $l$ is placed directly above $x_i$. Hence it is correct to add $l$ above $x_i$ on side $s$, and does \emph{not} declare $s$ as finished. 

\paragraph{\textbf{Case $\boldsymbol{U \neq \emptyset}$.}}The case $\boldsymbol{|L'| = 0}$ is identical to the corresponding subcase
when $U = \emptyset$. This means that in this case it is always correct to declare the side $s$ as finished and return $N$.

Consider now the case $\boldsymbol{|L'| \geq 2}$. Observe firstly that, if some taxon $l \in L'$ is placed directly above $x_i$, then all remaining taxa
in $L'$ \emph{must} be allocated to sides in $U$. To see why this is, note that for every $l' \in L'$ we have that $l' \rightarrow_{\cC}  x_i$.  So, if $l' \neq l$ is placed above $l$ on $s$ or on a side not in $U$, 
then, from Observation \ref{obs:sameSide} and \ref{obs:when_l'_impl_l} we would have that 
 $l' \rightarrow_{\cC}  l \rightarrow_{\cC}  x_i$, contradicting the fact that $l'$ is in $L'$. 
We only need to show that, if a valid completion for $N$ exists, then the set $\mathcal{N}$  contains a network for which there exists a valid completion.
Note that $\mathcal{N}$ contains (line 20) $N$, where $s$ is declared as finished, (lines 21-23) 
all possible networks obtained from $N$  by allocating all taxa in $L'$ to sides in $U$ and (lines 24-27)  all possible networks obtained from $N(l,s)$ by  allocating all taxa in $L' \setminus \{l\} $ to sides in $U$, iterating over all $l \in L'$.  The only
case that these three sets do not describe, is when every valid completion has a taxon $p \not \in L'$ directly above
$x_i$, but at least one taxon 
$l \in L'$ is  not mapped to $U$. But this implies, similarly to the case $|U|=0$, that $l \rightarrow_{\cC}  p \rightarrow_{\cC}  x_i$, so $l \not\in L'$, contradiction. Hence this case cannot happen, and  the three sets actually describe all possible
outcomes in this situation. So at least one of them will contain a network with a valid completion in the case $N$ does have a valid completion. 

 Consider now the case $\boldsymbol{|L'|=1}$. We begin with the subcase  $\boldsymbol{B(l) \neq \emptyset}$. Similar to previous arguments we know that, if we place $l$ (the only element in $L'$) directly above $x_i$, all taxa in $B(l)$ have to be allocated
to $U$. This holds because, from Observation \ref{obs:onlyCluster}, any cluster that contains $x_i$ but not $l$ is represented by the edge between the parents
of $l$ and $x_i$.
If $B(l) \neq \emptyset$, the set $\mathcal{N}$ is composed of 
(line 33) $N$, where $s$ is declared as finished, (lines 34-35) all possible networks obtained from $N$ by  allocating $l$ to a side in $U$ and (lines 36-38) all possible networks obtained from $N(l,s)$ by  allocating all taxa in $B(l)$  to sides in $U$.  
Observe that the only situation that these three guesses do not describe, is when  some
taxon $p \neq l$ is placed above $x_i$ and $l$ is not mapped to $U$. But in this case we would have that 
$l \rightarrow_{\cC}  p \rightarrow_{\cC}  x_i$, contradicting the fact that $l$ is in $L'$. So  $\mathcal{N}$ does again describe
all possible outcomes.

This leaves us with the very last subcase, $|L'|=1$ and $\boldsymbol{B(l)=\emptyset}$. The subcase when
$\boldsymbol{N(l,s)}$ {\bf does \emph{not} represent} $\boldsymbol{\cC}$ restricted to $\cX(N) \cup \{l\}$ is actually fairly straightforward. It is clear that
$l$ cannot be placed in this position in a valid completion. Hence the only two situations that 
line 41 and lines 42-43 do not 
describe, is when some element $p \neq l$ is placed directly above $x_i$, and $l$ is not mapped to $U$. But, as before, this implies that $l \rightarrow_{\cC}  p \rightarrow_{\cC}  x_i$, which as we have seen is not possible. So
the only remaining subcase is when $|L'|=1$, $B(l)=\emptyset$ and $\boldsymbol{N(l,s)}$ {\bf \emph{does} represent} $\boldsymbol{\cC}$ restricted to $\cX(N) \cup \{l\}$. 
Now, 
consider the network $N^{*}(l,s)$. Informally the dummy taxa in $N^{*}(l,s)$ act as ``placeholders'' for taxa that will only later
in the algorithm be mapped to $U$. We do not know exactly what these taxa will be, but we know that they will
definitely be there. Consider a cluster $C \in \cC^{*}$. If $N^{*}(l,s)$ does not represent $C$ then this must be because
of the dummy taxa, because we know that $N(l,s)$ did represent $\cC$ restricted to $\cX(N) \cup \{l\}$. Note that this holds irrespective of the
true identity of the dummy taxa. Hence, $C$ will never be represented by any completion of $N(l,s)$. For this
reason we conclude that, if $N^{*}(l,s)$  does not represent $\cC^{*}$, it is definitely correct to declare the side $s$ finished (line 49) or allocate $l$  to a side in $U$ (lines 50-51). 

Finally, suppose $N^{*}(l,s)$  does represent $\cC^{*}$. This is the flip-side of the previous argument. Whatever
the true identity of the dummy taxa, every valid completion of $N(l,s)$  will represent every cluster in $\cC^{*}$. 
Let $N'$ be an arbitrary valid completion of $N$ and denote by $N''$ the network obtained from $N'$ by moving $l$, wherever it is, just
above $x_i$. Now, as we did earlier we argue that in this case it is ``safe'' to put $l$ directly above $x_i$. Indeed, because $B(l)=\emptyset$, the only clusters that might not be represented in $N''$ are clusters in $\cC^{*}$.
But we have shown that when $l$ is placed directly above $x_i$ all the clusters in $\cC^{*}$ are represented
regardless of how we complete the rest of the network. Secondly, we argue just as before that moving $l$ in this
way cannot alter the side types. So if we choose $l$ as $x_{i+1}$ there must still exist a valid completion.  This concludes the proof of the lemma. 
\end{proof}

\begin{lemma}
\label{lemma:step1}
Let $\cC$ be a separating set of clusters on $\cX$ and let $k$ be the first integer for which a level-$k$ network representing $\cC$ exists. 
Let $N$ be an incomplete network such that 
its underlying generator $G$ and set of side guesses $S_G$ are such that $(G,S_G)$ is side-minimal w.r.t. $\cC$ and $k$, and let $s$ be an active side of $N$ that contains only a single taxon. 
Then, if a valid completion for $N$ exists, Algorithm \ref{algo:completeSide}  computes  in  $f(k) \cdot poly(n)$ time a set of (incomplete) networks $\mathcal{N}$ for which $\boldsymbol{s}$ {\bf is a finished side},  such that $\mathcal{N}$  contains at least one  network  for which there  exists  a valid completion.  
\end{lemma}
\begin{proof}
The correctness follows from Lemma \ref{lemma:subroutineIsCorrect}.
We now prove the running time.

First, note that the size of the set $\mathcal{N}$ returned by Algorithm \ref{algo:addOnSide} is bounded by $f(|U|)$. This is evident for  the sets $\mathcal{N}$  constructed  on lines
34-35, 42-43 and 50-51 
 but it holds also for the sets $\mathcal{N}'$  constructed respectively on lines 
 21-23, 24-27 and 36-38, 
 since these sets are constructed only if, respectively,  $|L'| \leq |U|$, $|L'|-1 \leq |U|$ or $|B(l)| \leq |U|$. 
 Since in all other cases $|\mathcal{N}|=1$, the size of the set $\mathcal{N}$ returned by Algorithm \ref{algo:addOnSide}  is indeed bounded by $f(|U|)$. Moreover, from Proposition \ref{prop:checkReprCluster}, it follows that the running time of Algorithm \ref{algo:addOnSide} is $f(k) \cdot poly(n)$.  

Second, note  that, each time that Algorithm \ref{algo:addOnSide} returns a set of networks $\mathcal{N}$ such that  $|\mathcal{N}|>1$, $|U|$ decreases \cs{or $s$ is declared as finished}.
Additionally, when $U = \emptyset$ \cs{and $|\mathcal{N}=1|$}, Algorithm \ref{algo:addOnSide}  returns only one network per call and we have at most $O(n)$ of these calls (because either $s$ is declared finished or a new taxon is added to $s$).
 
 Since the number of sides in a generator is bounded by $f(k)$ and $U$ is a subset of the short sides of the generator (which follows from the fact that all long sides reachable from $s$ are assumed
to be finished), we have that $|U|$ is bounded by $f(k)$. Thus, the running time of  Algorithm \ref{algo:completeSide} is bounded by $f(k) \cdot poly(n)$. 
\end{proof}

\begin{algorithm}[]
$ \mathcal{N} \gets{N}$\;
\While{there exists $N \in \mathcal{N}$ such that $s$ is not finished in $N$}{
		$\mathcal{N} \gets{} \mathcal{N}  \setminus N$\;
		$\mathcal{N} \gets{}$ $\mathcal{N}$ $\cup$ addOnSide$(N, s)$\;
}

\caption{
completeSide($N,s$)\label{algo:completeSide}
}
\end{algorithm}

We will subsequently use the term  \emph{lowest side} to denote a long side that does not yet have all its taxa (i.e. an unfinished long side), and such that there is no other long side $s' \neq s$ with this property that is reachable from $s$. 

The following lemma is basically the fixed parameter tractable version of Lemma 3 from \cite{elusiveness}:
\begin{lemma}
\label{lem:core}
Let $\cC$ be a 
separating  set of clusters
on $\cX$. Then, for every fixed
$k \geq 0$, Algorithm \ref{algo:core} determines  whether a level-$k$ network
exists that represents $\cC$, and if so, constructs such a network in time $f(k) \cdot poly(n)$.
\end{lemma}
\begin{proof}

Algorithm \ref{algo:core} starts by choosing a  level-$k$ generator $G$ and a set of side guesses $S_G$.   Note that (see lines 1-2) generators and sets of guesses are analyzed in such a way that generators with a smaller number of sides and sets of side guesses  with a smaller number of long sides, and (to further break ties) short sides, are analyzed first (this is the meaning of the expression ``{in increasing side order}'').  
This implies that the side-minimal pair $(G,S_G)$, if any exists, is analyzed before any other pair $(G',S_G')$ for which a \cs{valid} completion exists.  
This is done to be able to apply Lemma \ref{lemma:step1}. 

Then (see lines 4-18), the algorithm constructs a set of \emph{complete} networks,  i.e. simple level-$k$ networks where each short side has received a single taxon and each long side at least two,  and returns the first of them that represents $\cC$, if any exists.  On line 6, $\cX(s^{-})$ denotes the set of all taxa in $\cX$ that are candidates to be the first taxon on side $s$, which we call
$s^{-}$. We will discuss this set in more detail shortly.
\RestyleAlgo{ruled}
\begin{algorithm}
\KwData{A separating  set of clusters $ \cC$ on $\cX$}
\KwResult{A level-$k$ network representing $\cC$, if any exists}
 \ForEach{level-$k$ generator $G$ in  increasing side order }{
	 \ForEach{set of side guesses $S_G$ in  increasing side order}{
		$\mathcal{N} \gets{(G,S_G)}$\;
		 \While{there exists $N \in \mathcal{N}$ such that $N$ contains a lowest side $s$}{
			$\mathcal{N} \gets{} \mathcal{N}  \setminus N$\;
				\ForEach{$s^{-} \in \cX(s^{-})$}{
					 $\mathcal{N'}\gets{}$ completeSide($N(s^{-},s),s$)\;
			 		$\mathcal{N''} \gets{}$ the networks in $\mathcal{N'}$ where $s$ contains more than one taxon\;
					\ForEach{$N \in \mathcal{N}''$}{ 
						collapse all  taxa on side $s$ into a single meta-taxon $S$ and adjust the cluster set accordingly\;
					}
			 		$\mathcal{N} \gets  \mathcal{N}  \cup  \mathcal{N''}$\; 	
				}
		}
		\If{$|\mathcal{N}| >0 $}{
			 \While{there exists $N \in \mathcal{N}$}{
				 de-collapse the collapsed sides\; 
				 $\mathcal{N}\cs{'} \gets{}$ the set of networks obtainable from $N$ by allocating the taxa in $\cX \setminus \cX(N)$ to  any short sides that have not yet been allocated a taxon\;
				\If{there is a network $N\cs{'} \in \mathcal{N\cs{'}}$ representing $\cC$}{
					\Return $N\cs{'}$\;		
				}
				$\mathcal{N} \gets{} \mathcal{N}  \setminus N$\;
			}
		}			
	}
}
\Return $\emptyset$\;	
\caption{ComputeLevel-$k$($\cC$)\label{algo:core}}
\end{algorithm}

Note that lines $10$ and $14$ of the algorithm  are only a technical step. Indeed, when we declare a side $s$ as finished, we assume that we will never alter that side again. Hence it does not change the analysis if
we collapse all the taxa on side $s$ into a single meta-taxon. That is, if we have decided  that the taxa on the side $s$ are $s^{-}, x_1, \dots, x_{l}$ we simply replace all these taxa by a single new taxon $S$ and replace
$s^{-}, x_1, \dots, x_{l}$ by $S$ in any clusters in $\cC$ that they appear in (line 10). This collapsing step is simply a convenience to ensure that the set of sides reachable from the current lowest side are always empty or short sides. This will be helpful when proving the running time of Algorithm \ref{algo:core}, see below. 
When we are finished allocating all the taxa in $\cX$ and are ready to check whether the resulting final network represents $\cC$ we can simply de-collapse all the $S$ i.e. ``unfold'' all the long sides that we have collapsed (line 14). 
This means that the correctness of Algorithm \ref{algo:core} follows by Observation \ref{obs:searchingSpace} and Lemma \ref{lemma:step1}.

We now need to prove the correctness of the  running time. First, note that the number of pairs $(G,S_G)$ to consider is bounded by $f(k)$ since both the number of generators and the number of sides per generator are bounded by $f(k)$. 

We now need to prove that the size of $\cX(s^{-})$ is at most $f(k)$ for all sides $s$ i.e. that the number of taxa that might be the first taxon $s^{-}$ on side $s$, is not too big. So let $s^{-}$ be
any taxon which can fulfil this role, and let $x$ be the taxon directly above $s^{-}$ on side $s$. (The taxon $x$ must exist because we assume that $s$ is long).
Clearly, $x \rightarrow_{\cC} s^{-}$.
By line 10, 
we have that 
 the only sides
reachable from side $s$ are short and empty sides. Moreover, we know from Observation \ref{obs:size2Incomp} that, because $\cC$ is separating, there is some non-singleton cluster $C \in \cC$ such that $s^{-} \in C$ but
$x \not \in C$. By Observation \ref{obs:onlyCluster}, such a cluster $C$ has to be represented by 
the edge $e$ between the parents of $x$ and $s^{-}$. Now, any cluster represented by $e$  
can only contain taxa that are reachable from $e$ by a directed path. The only sides that are reachable from side $s$ are short and empty sides, so the cluster $C$ can only contain at most $f(k)$ taxa (because there are at most $f(k)$ short sides). So we know
that $s^{-}$ is in some cluster $C$, and that $C$ is ``small'' in the sense that its size is bounded above by $f(k)$. So if we take all ``small'' clusters, and let $\cX(s^{-})$ be their union, we know that we could simply try
taking every element in $\cX(s^{-})$ and guessing that it is equal to $s^{-}$. To ensure that we do not use too many guesses, we have to show that $|\cX(s^{-})|$ is bounded by $f(k)$. To see that this holds, consider the question: how many clusters in $\cC$ contain at most $c$ taxa for a constant $c$? Observe that on every long side only the $c$ taxa furthest
away from the root are \emph{potentially} in such clusters. Any taxon closer to the root on a long side cannot possibly be in a cluster of size at most $c$, because if it is in a cluster then so are at least $c$ other taxa too.
Hence there are at most $f(k)$ taxa that can be involved in ``small'' clusters: the taxa on the short sides and the taxa at the bottom of the long sides.
So we have that $|\cX(s^{-})|$ is bounded by $f(k)$ and we can guess $s^{-}$ with at most $f(k)$ guesses.

The collapsing and de-collapsing steps (lines 10 and 14) can be  done in $f(k) \cdot poly(n)$ time, as well as completing each side $s$ (line 7), by Lemma \ref{lemma:step1}. Moreover,  by Proposition \ref{prop:checkReprCluster}, 
 also checking whether $N$ represents $\cC$ takes $f(k) \cdot poly(n)$ time. Additionally,  the allocation of remaining taxa to the unfinished short sides (line 15) takes a time bounded by $f(k)$. Indeed, if we have correctly guessed all the taxa on the long sides, then there will be at most $f(k)$ taxa left that have not yet been assigned
to a side, and these will correspond to taxa on (a subset of) the short sides. (Some of the short sides might already have been allocated taxa by Algorithm \ref{algo:addOnSide}). 

This means that, if the size of $\mathcal{N}$ is bounded by $f(k)$, then  the entire  algorithm can be executed in $f(k) \cdot poly(n)$ time. And this is indeed the case, since $|\cX(s^{-})|$ and (summing over all iterations) the total number of lowest sides are  bounded by $f(k)$ and each time that a side is completed, the number of unfinished long sides decreases by 1.


\end{proof}

\subsection{\emph{From simple networks to general networks}}
\label{subsec:generalFPT}

To prove the fixed parameter tractability of constructing general level-$k$ networks, we need to introduce a few other concepts.
The most important is the concept of a \emph{decomposable network}.

\begin{definition}\label{def:decomposable}
Let $\cC$ be a set of clusters  on a taxon set $\cX$ with incompatibility graph $IG(\cC)$ and let $N$ be a phylogenetic network 
that represents $\cC$. $N$ is said to be \emph{decomposable} w.r.t. $\cC$ if and only if there exists a cluster-to-edge mapping $\alpha: \cC \rightarrow E(N)$  such that, for
any two clusters $C_1,C_2\in \cC$, 
$C_1$ and $C_2$ lie in the same connected component of $IG(\cC)$ if and only if 
the two tree edges $\alpha(C_1)$ and $\alpha(C_2)$ that represent $C_1$ and
$C_2$ are contained in the same biconnected component of $N$. 
\end{definition}

The following observation is straightforward:

\begin{observation}\label{osb:maxBiconnectedComponent}
Let $\cC$ be a set of clusters  on a taxon set $\cX$ with incompatibility graph $IG(\cC)$ and let $N$ be a phylogenetic network $N$ representing $\cC$ that is decomposable w.r.t. $\cC$. Then, we have that the number of biconnected components of $N$  is equal to the number of connected components of  $IG(\cC)$. 
\end{observation}

Let $\cC$ be a set of clusters on $\cX$ with incompatibility graph $IG(\cC)$.
The set of \emph{backbone clusters} associated with $\cC$ is defined as
\[
	B(\cC)=\{\cX(\cC')\mid \cC' \mbox{~is a connected component of~}IG(\cC)\},
\]
where $\cX(\cC')=\cup_{C\in\cC'}C$ denotes the set of all taxa  in $\cC'$. 
Since the set $B(\cC)$ is compatible \cite{HusonRuppScornavacca10}, we have the following result:
\begin{proposition}
\label{prop:checkReprClusterGeneral}
Given a decomposable level-$k$ network  $N$ representing  a set of clusters  $\cC$ on $\cX$, $N$ can contain at most $2^{k+3}\cdot(n-1)^2$ clusters. 
\end{proposition}
\begin{proof}
The fact that $B(\cC)$ is compatible  ensures that the size of $B(\cC)$ is at most $2(n-1)$. In the following we will prove that the number of connected components of $IG(\cC)$ is at most $4(n-1)$. To prove this, we show that it is impossible to have two non-trivial connected components of $IG(\cC)$ (i.e. two connected components containing more than once cluster each), say $\bar{\cC}$ and $\bar{\cC}'$, such that $\cX(\bar{\cC})=\cX(\bar{\cC}')$. For the sake of contradiction, let us suppose that two such components $\bar{\cC}$ and $\bar{\cC}'$ exist. 
 Let $C_1 \in \bar{\cC}$ and $C^{'}_1 \in \bar{\cC}'$ be two clusters such that $C_1 \cap C'_1 \neq \emptyset$. Since $C_1$ and $C'_1$ are compatible, we can suppose w.l.o.g that $C'_1 \subset C_1$. Let $C'_2$ be another cluster of $\bar{\cC}'$ incompatible with $C'_1$ ($C'_2$ exists because $\bar{\cC}'$ is not trivial). Then, since $C'_1 \cap C'_2 \neq \emptyset$ we have that $C_1 \cap C'_2 \neq \emptyset$. But $C'_2$ cannot be a superset of $C_1$ so we have that $C'_2 \subset C_1$. Reiterating this reasoning we obtain that $\cX(\bar{\cC}') = C_1$.
Since  $\bar{\cC}$ is not trivial, there exists another cluster $C_2$ in $\bar{\cC}$ that is incompatible with $C_1$. So there exists at least one taxon in $\cX(\bar{\cC})$ that is not in $C_1$ and  we cannot have that $\cX(\bar{\cC})=\cX(\bar{\cC}')$, contradiction. This means that each {non-singleton} backbone cluster can correspond to two connected components, one trivial and one not. Then we have at most $4(n-1)$ connected components in $IG(\cC)$. By Observation \ref{osb:maxBiconnectedComponent}, this implies that $N$ has at most $4(n-1)$ biconnected components. 

We now prove that each biconnected component $B$ of $N$ can represent at most $2^{k+1}(n-1)$ clusters. 
To see that, let us denote by $V'$ the set of nodes of $N$ that are not in $B$ but whose parents are in $B$ and, for each $v \in V'$, denote by $\cX(v)$ the set of all leaves in $N$ that are reachable by directed paths from $v$. It is easy to see that the set  $V'$ has a particularity: for each node $v \in V'$ we have that, no matter which switching $T_N$ is chosen, there exists a path in the
switching between $v$ and each taxon $u \in \cX(v)$. {Indeed, if this was not true, we will have that $v$ has to be in $B$, a contradiction}.
Then the network $N$ can be modified in the following way: \cs{For each node $v\in V'$,  label it with the set $\cX(v)$ and delete all the outgoing edges of $v$. 
  Let $N'$ be the rooted phylogenetic network rooted at the root of the biconnected component $B$. 
Because of the  peculiarity  of the nodes in  $V'$, $N'$ represents  the same
cluster set than $B$ in $N$. }
%
With a line of reasoning similar to that used in Proposition \ref{prop:checkReprCluster}, it is easy to see that $B$ can represent at most $2^{k+1}(n-1)$ clusters in $N'$, and thus
also in $N$. This concludes the proof.
\end{proof}

The following theorem ensures that we can focus on decomposable level-$k$ networks:
\begin{theorem}[\cite{cass}]\label{thm:level-$k$-decomp}
Let $\cC$ be a set of clusters. If there exists a level-$k$ network representing $\cC$, then there also exists such a network that is decomposable w.r.t. $\cC$.
\end{theorem}

We can now prove the main result of the section:


\begin{theorem}
\label{thm:core}
Let $\cC$ be a set of clusters
on $\cX$. Then, for every fixed
$k \geq 0$, it is possible to determine in time $f(k) \cdot poly(n)$ whether a level-$k$ network
exists that represents $\cC$, and if so to construct such a network.
\end{theorem}
\begin{proof}
By Theorem \ref{thm:level-$k$-decomp}, we know that we can construct a decomposable level-$k$ network using a divide-and-conquer strategy. A possible approach is described in Section 8.2 of \cite{HusonRuppScornavacca10}. This approach divides $\cC$ in $g$ subsets, where $g$ is the number of connected components of the incompatibility graph. Then, each subset $\cC_i$ is made separating w.r.t. $\cX(\cC_i)$ by merging 
every subset of $\cX(\cC_i)$ that is compatible with 
$\cC_i$ 
(see \cite{HusonRuppScornavacca10} for more details). Then a \emph{local} network is computed for each $\cC_i$ and finally all the networks are merged together in a \emph{global} level-$k$ network representing $\cC$. Since by Proposition \ref{prop:checkReprClusterGeneral} we can only have a $f(k) \cdot poly(n)$  number of clusters and a $pol(n)$ number of connected components in $IG(\cC)$, it is easy to see that the merging of the taxa in each $\cC_i$ and the merging of all partial networks into the global one can be conducted in $f(k) \cdot poly(n)$ time. Moreover, since each subproblem $\cC_i$ is separating, from Lemma \ref{lem:core} we have that constructing 
each local network takes $f(k) \cdot poly(n)$ time. This concludes the proof. 
\end{proof}

\section{Minimizing reticulation number is fixed parameter tractable}
\label{sec:retnumFPT}

The aim of this section is to show that reticulation number minimization is  fixed parameter tractable. 
As pointed out for level minimization at the beginning of Section \ref{sec:levelFPT}, it
is sufficient to prove that, given a  set of clusters $\cC$ on taxon set $\cX$, we can construct a phylogenetic network representing $\cC$ with reticulation number $r$ (if any exists)  in time at most $f(r) \cdot poly(n)$.

To show the main result of this section we will introduce the concepts of \emph{ST-collapsed cluster sets} and of \emph{$r$-reticulation generators}. We will then prove that  all the results and  algorithms used in the previous section to prove that constructing \emph{simple} level-$k$ networks is fixed parameter
tractable,  hold not only for separating cluster sets and level-$k$ generators but also for ST-collapsed cluster sets and $r$-reticulation generators. The main difficulty is to show that 
several key utility results still hold, since the others results do not exploit the biconnectedness of simple level-$k$ generators. For ease of reading we will refer to the extended versions of these results using their original name followed by the term ``(extended)" (e.g. ``Proposition 1"  becomes ``Proposition 1 (extended)").\\

%

\begin{observation}
\label{obs:bound}
Let $N$ be a network on $\cX$ with reticulation number $r$. Then $N$ represents at most $2^{r+1}(n-1)$ clusters.
\end{observation}
\begin{proof}
From Lemma \ref{lem:binaryOK} we may assume without loss of generality that $N$ is binary. A binary network with reticulation number $r$ contains
exactly $r$ reticulation nodes. Hence $N$ displays at most  $2^{r}$ trees, and each tree represents at most $2(n-1)$ clusters (because
a rooted tree on $n$ taxa contains at most $2(n-1)$ edges). 
\end{proof}

We can thus
henceforth assume that $|\cC| \leq 2^{r+1}(n-1)$ i.e. that $\cC$ contains at most $f(r) \cdot poly(n)$ clusters.  
Then we have that the following holds:
\begin{prop1bis}
\label{prop:checkReprClusterBis}
Let $N$ be a network on $\cX$ with reticulation number at most $r$. Then,
given a cluster set $\cC$, we can check in time $f(r) \cdot poly(n)$ whether
$N$ represents $\cC$.
\end{prop1bis}


Given a set of taxa~$S\subseteq\mathcal{X}$,  
 we use~$\mathcal{C}\setminus S$ to denote the result of removing all elements of~$S$ from each cluster
in~$\mathcal{C}$ and we use~$\mathcal{C}|S$ to denote~${\mathcal{C}\setminus (\mathcal{X}\setminus S)}$ (i.e. the restriction of~$\mathcal{C}$ to~$S$). We say that a set~$S \subseteq \mathcal{X}$ 
is an \emph{ST-set} with respect to $\mathcal{C}$, if~$S$ is compatible with $\mathcal{C}$ and any two clusters~$C_1,C_2\in\mathcal{C}|S$ are 
compatible \cite{elusiveness}. 
(We say that an ST-set $S$ is \emph{trivial} if $S= \emptyset $ or $S= \cX$).
 An ST-set~$S$ is 
\emph{maximal} if there is no ST-set~$S'$ with~$S \subset S'$. The following results from \cite{elusiveness} will be very useful:

\begin{corollary}
\label{cor:mostn}
\cite{elusiveness}. Let $\mathcal{C}$ be a set of clusters on $\cX$. Then there are at most $n$ maximal ST-sets with respect to $\cC$, they are uniquely
defined and they partition $\cX$.
\end{corollary}

\begin{lemma}
\label{lem:polymax}
\cite{elusiveness} The maximal ST-sets of a set of clusters $\cC$ on $\cX$ can be computed in polynomial time.
\end{lemma}

Here ``polynomial time'' means $poly(n,|\cC|)$, but given that the size of $\cC$ is at most $f(r) \cdot poly(n)$ it follows that the maximal ST-sets of $\cC$ can all be computed
in time $f(r) \cdot poly(n)$. 

The following corollary says, essentially, that if we want to construct networks with minimum reticulation number then it is safe to assume that
each maximum ST-set corresponds to (the taxa in) a subtree that is attached to the main network via a cut-edge.

\begin{corollary}
\label{cor:maxundercut}
\cite{elusiveness} Let $N$ be a network that represents a set of clusters $\cC$. There exists a network $N'$
such that $N'$ represents $\cC$, $r(N') \leq r(N)$, $l(N') \leq l(N)$ and all maximal ST-sets (with respect to $\cC$)
are below cut-edges. 
\end{corollary}

Let $\cS = \{S_1, \ldots, S_m\}$ be the set of maximal ST-sets of $\cC$. We construct a new cluster set $\cC'$ from $\cC$ as follows. \cs{For
each $S_j \in \cS$, and for each  cluster $C$ in $\cC$ such that $S_j \cap C := S\neq\emptyset$, we replace the set $S$ in $C$ by}  
the new taxon $s_j$. In other words we ``collapse'' all taxa in each maximal ST-set into a single new taxon that represents that ST-set. We say that $\cC'$ is the \emph{ST-collapsed} version of $\cC$. Note that a separating cluster set $\cC$ is necessarily  ST-collapsed   but the opposite implication does not hold. For example $\cC=\{
\{a,b\}, \{b,c\}, \{a,b,c,d\}, \{d,e\}\}$ on $\cX=\{a,b,c,d,e\}$ is ST-collapsed but not separating because $\{a,b,c\}$ is compatible with $\cC$.

\begin{obs5bis}
\label{obs:size2IncompBis}
Let $\cC$ be a ST-collapsed cluster set on $\cX$. Then every size-2 subset of  $\cX$ is incompatible with $\cC$. 
\end{obs5bis}
\begin{proof}
Suppose that this is not true and there exists a size-2 subset of  $\cX$, say $A$, that is compatible with $\mathcal{C}$. Since any two clusters~$C_1,C_2\in\mathcal{C}|A$ are necessarily compatible, $A$ is a ST-set, contradicting the fact that  $\cC$ is ST-collapsed. 
\end{proof}

\begin{lemma}
\label{lem:stcollapse}
Let $\cC$ be a cluster set on $\cX$, and let $\cC'$ be the ST-collapsed version of $\cC$. Then any network $N'$ that represents $\cC'$ can be transformed into a network
$N$ that represents $\cC$, such that $r(N)=r(N')$ in $f(r(N)) \cdot poly(|\cX|)$ time.
\end{lemma}
\begin{proof}
Let $\cS = \{S_1, \ldots, S_m\}$ be the set of maximal ST-sets of $\cC$. For each $S_j \in \cS$ we replace the taxon $s_j$ in $N'$ with the tree on taxon set $S_j$ that represents
exactly the set of clusters $\cC|S_j$. 
\end{proof}

\begin{corollary}
\label{cor:STisminimal}
Let $\cC$ be a cluster set on $\cX$, and let $\cC'$ be the ST-collapsed version of $\cC$. Then $r(\cC') = r(\cC)$.
\end{corollary}
\begin{proof}
Lemma \ref{lem:stcollapse} tells us that $r(\cC) \leq r(\cC')$. To see that $r(\cC') \leq r(\cC)$, observe that Corollary \ref{cor:maxundercut} allows us
to assume the existence of a network $N$ with reticulation number $r(\cC)$ such that all the maximal ST-sets of $\cC$ are below cut-edges in $N$. If,
for each maximal ST-set $S_j$ of $\cC$, we replace the subtree corresponding to $S_j$ with a single taxon $s_j$, we obtain a network with reticulation
number at most $r(\cC)$ which represents $\cC'$. 
\end{proof}

Combining the fact  the  transformation described in the proof of Lemma \ref{lem:stcollapse} can  be executed in time $f(r) \cdot poly(n)$ 
with  Lemma \ref{lem:stcollapse} and Corollary \ref{cor:STisminimal} we may thus henceforth restrict our attention to ST-collapsed cluster sets. Networks that represent ST-collapsed cluster sets have a rather
restricted topology, as the following lemma shows.

\begin{lemma}
\label{lem:nocut}
Let $\cC$ be an ST-collapsed cluster set on $\cX$, and let $N$ be a binary network that represents $\cC$. Then
it follows that, for each cut edge $(u,v)$ of $N$, either $v$ is a leaf labelled by a taxon from $\cX$, or there
is a directed path starting from $v$ that can reach a reticulation node.
\end{lemma}
\begin{proof}
Let $\cX(v) \subseteq \cX$ be the set of taxa reachable from $v$ by directed paths. If $|\cX(v)| \geq 2$, but
there are no reticulation nodes reachable by directed paths from $v$, then the subnetwork rooted at $v$
is actually a tree with taxon set $\cX(v)$, meaning that $\cX(v)$ is an ST-set of cardinality 2
or higher. This violates the ST-collapsed assumption, giving a contradiction. If $|\cX(v)|=1$ then it follows that either $v$ is a leaf labelled by a taxon, or (due to the fact that $N$ is binary and contains no nodes with indegree and outdegree both equal to 1) at least one reticulation node is reachable from $v$ by a directed path. 
\end{proof}

We are now (finally) ready to define a $r$-\emph{reticulation generator}. This is very closely related to the level-$k$ generator discussed in Section \ref{sec:levelFPT}. The only significant difference is that $r$-reticulation generators do not have to be biconnected, and
(for technical reasons) the inclusion of a ``fake root''.

\begin{definition}
An $r$-\emph{reticulation generator} is a directed acyclic
multigraph, which has a single node of indegree
0, called the \emph{fake root}, and this has outdegree 1; precisely $r$ reticulation nodes (indegree 2
and outdegree at most 1), and apart from that only nodes
of indegree 1 and outdegree 2.
\end{definition}

Note that this definition implies that a  $r$-reticulation generator cannot contain any leaf. As in the case of level-$k$ generator,  nodes with
indegree 2 and outdegree 0 as well as all edges are  
called \emph{sides}.
Figure \ref{fig:rgen} shows the single 1-reticulation generator and the seven 2-reticulation generators.

\begin{sloppypar}
\begin{lemma}
\label{lem:numbRetGen}
There are at most $f(r)$ $r$-reticulation generators and each $r$-reticulation generator contains at most
$f(r)$ sides.
\end{lemma}
\end{sloppypar}
\begin{proof}
In Lemma 1 of \cite{reflections} it is proven that a level-$k$ generator has
at most $3k-1$ vertices and at most $4k-2$ edges. The proof there does not exploit the biconnectedness
of level-$k$ generators, so - with the exception of the fake root - also holds for
$r$-reticulation generators. By adding 1 to both the vertex and edge upper bounds to account for the
fake root we come to upper bounds of $3r$ and $4r-1$ respectively. Hence we obtain $7r-1$ as a crude
upper bound on the number of sides in a generator. To see that there are at most $f(r)$ $r$-reticulation
generators observe that between any pair of nodes $u$ and
$v$ in the generator there is either no edge, an edge from $u$ to $v$ (or from $v$ to $u$), or a multi-edge from $u$ to $v$
(or from $v$ to $u$). Hence there are at most $5^{(3r)^2}$ $r$-reticulation generators.  
\end{proof}

\begin{definition}
\label{def:Nl}
The set $\mathcal{\hat{N}}^{r}$ (for $r \geq 1$) is defined as the set of all binary networks that can be constructed by choosing some $r$-reticulation
generator $G$, then applying the {leaf hanging} transformation described in Definition \ref{def:Nk} {and} finally 
deleting the fake root (i.e. the single vertex with indegree 0 and outdegree 1) and its incident edge.
\end{definition}

\begin{lemma4bis}
\label{lem:genIsSufficient_bis}
Let $\cC$ be an ST-collapsed set of clusters on $\cX$, such that $r(\cC) \geq 1$. Then there exists a network $N$ in $\mathcal{\hat{N}}^{r(\mathcal{C})}$ such that
$N$ represents $\cC$.
\end{lemma4bis}
\begin{proof}
Let $N$ be any binary network with reticulation number $r(\cC)$ such that
$N$ represents $\cC$. We show how applying the reverse of the transformation
described in Definition \ref{def:Nl} to $N$ will give some $r(\cC)$-reticulation
generator $G$. The lemma will then follow. We begin by adding a fake root to $N$ i.e. a new vertex $u'$ and an edge from $u'$ to
the root of $N$. (This is the inverse of deleting the fake root). We then delete all the leaves in $N$. Any nodes that are
created with indegree 2 and outdegree 0 we leave as they are (this is the inverse of step 3 of Definition \ref{def:Nk}). Nodes with indegree 1 and outdegree 0 cannot be created, because this would require that there exists a node $v$ in $N$ which has indegree 1 and outdegree 2 such that both its children are leaves labelled by taxa. But this would mean that $v$ is the head of a cut-edge $e$ where $e$ violates the condition described in Lemma \ref{lem:nocut}. Now, consider the nodes that have been created with indegree and outdegree both equal to 1. Let $u$ be any such node, and let $U = \{u\}$. Whenever $U$ contains a node $u$ whose unique parent $p(u)$ also has indegree and outdegree both equal to 1, add $p(u)$ to $U$. Whenever $U$ contains a node $u$ whose unique child $c(u)$ also has indegree and outdegree both equal to
1, add $c(u)$ to $U$. We continue expanding $U$ this way until it cannot grow anymore. Clearly $U$ stops growing at the point that $U$ contains two
nodes $u_{top}$ and $u_{bottom}$ (where possibly $u_{top} = u_{bottom}$) such that the parent of $u_{top}$ (respectively, child of $u_{bottom}$)
does not have indegree and outdegree both equal to 1. We suppress all the nodes in $U$, in the usual sense. Note, crucially, that this does not affect the indegree or outdegree of the parent of $u_{top}$ or the child of $u_{bottom}$. While $N$ still contains nodes of indegree 1 and outdegree 1 we repeat the above process, until none are left; this is the inverse of steps 1 and 2 of Definition \ref{def:Nk}. (Note that this process might create multi-edges, but because it leaves the indegree and outdegree of unsuppressed nodes intact, there will be at most two edges between any two nodes). Now, let $G$ be the resulting structure. Observe that the reticulation number of $G$
is the same as $N$, that every node in $G$ with indegree 2 has outdegree 0 or 1, that $G$ contains a single fake root, that all nodes in $G$ with indegree 1 have outdegree 2, and that $G$ contains no leaves. We conclude that $G$ is a $r$-reticulation generator and that we could have constructed $N$ by applying the transformation described in Definition \ref{def:Nl} to $G$. 
\end{proof}


\begin{figure}
\label{fig:rgen}
\includegraphics[scale=0.2]{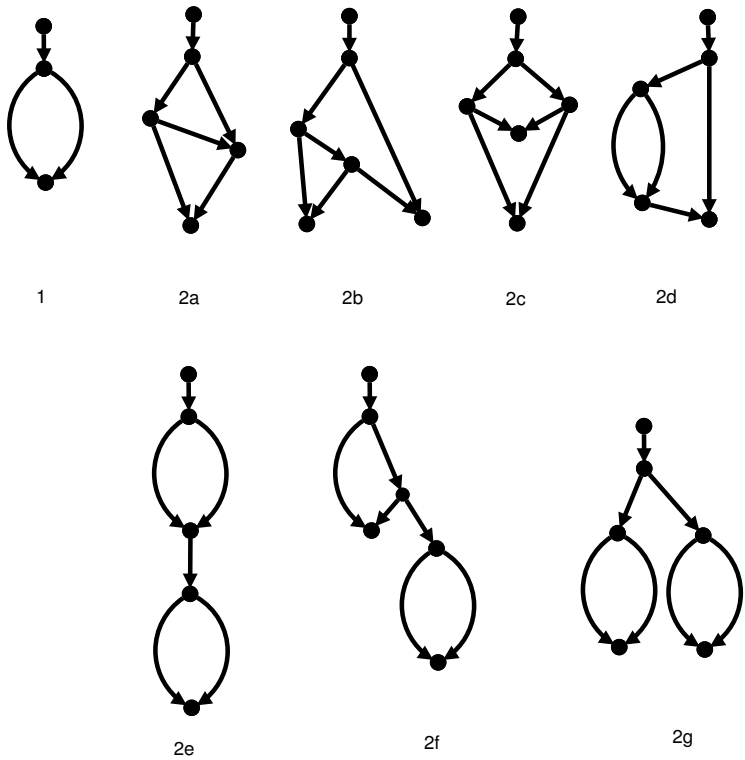}
\caption{The single 1-reticulation generator and the seven 2-reticulation generators.}
\end{figure}

\begin{prop2bis}
\label{prop:noCyclesRightarrowBis}
Given a ST-collapsed cluster set   $\cC$ on $\cX$ and an ordered set of distinct taxa of $\cX$ $(x_1, \dots, x_k)$ such that  $k\geq 2$ and $x_i \rightarrow_{\cC}  x_{i+1}$ for $ 1 \leq i \leq (k-1)$. Then $x_k \not\rightarrow_{\cC}  x_1$.
\end{prop2bis}
\begin{proof}
If $x_i \rightarrow_{\cC}  x_{i+1}$ for $ 1 \leq i \leq (k-1)$ and $x_k \rightarrow_{\cC}  x_1$, this means that the set $\cX'= \cup_{i=1}^{k}  x_i$ is compatible with $\cC$. Moreover,  every non-singleton cluster that contains one element of $\cX'$, contains
them all. So every non-singleton cluster $C\in \mathcal{C}$ is either disjoint from $\cX'$, or contains it, from which we conclude
that  any two clusters~$C_1,C_2\in\mathcal{C}|\cX'$ are 
compatible. So $\cX'$ is a ST-set and we  have  a contradiction. 
\end{proof}

Since the number of $r$-reticulation generators  and the number of sides in a generator is bounded by $f(r)$ from Lemma \ref{lem:numbRetGen}, and we have extended
several critical utility results to also apply to ST-collapsed cluster sets and $r$-reticulation generators, we have the following:
\begin{lemma7bis}
\label{lem:coreBis}
Let $\cC$ be a ST-collapsed  set of clusters on $\cX$. Then, for every fixed
$r \geq 0$, Algorithm \ref{algo:core} determines  whether a network  that represents $\cC$ with reticulation number equal to $r$ 
exists, and if so, constructs such a network in time $f(r) \cdot poly(n)$.
\end{lemma7bis}

%


From Lemma \ref{lem:stcollapse} and Corollary \ref{cor:STisminimal}, we may finally conclude the following.

\begin{theorem}
\label{thm:retics}
Let $\cC$ be a set of clusters on $\cX$. Then, for every fixed
$r \geq 0$, it is possible to determine in time $f(r) \cdot poly(n)$ whether a network that represents $\cC$ with reticulation number at most
$r$ exists.
\end{theorem}

\section{Conclusions and open problems}
In this article we have shown that, under the softwired cluster model of phylogenetic networks, constructing networks with minimum reticulation number (respectively, level) is fixed parameter tractable where the reticulation number (respectively, level) is the parameter. The obvious problem with the algorithms in this article is that
the part of the running time that depends only on the parameter is massively exponential. This contrasts with fixed parameter tractable algorithms for combining \emph{two trees}
into a phylogenetic network. In this literature the dependence on the parameter is more modest: for example in \cite{quantifyingreticulation} an algorithm for two binary trees with  running time of 
$O((14r)^{r} \cdot n^3)$  is given (where $n=|\cX|$ and $r$ is the hybridization number of the two trees). However, the two tree case is rather special \cite{twotrees,elusiveness} and it is likely that, as the number of trees increases, the dependence on the parameter will also increase dramatically. Relatedly, there is still no fixed parameter tractable algorithm for combining an arbitrary set of trees into a phylogenetic network using a minimum number of reticulations. Could the ideas presented in this article - in particular, the use of generators - offer a theoretical route to this result?


\bibliographystyle{plain}      
\bibliography{article_arxiv}

\begin{thebibliography}{10}

\bibitem{bordewich2}
M.~Bordewich, S.~Linz, K.~St. John, and C.~Semple.
\newblock A reduction algorithm for computing the hybridization number of two
  trees.
\newblock {\em Evolutionary Bioinformatics}, 3:86--98, 2007.

\bibitem{sempbordfpt2007}
M.~Bordewich and C.~Semple.
\newblock Computing the hybridization number of two phylogenetic trees is
  fixed-parameter tractable.
\newblock {\em IEEE/ACM Transactions on Computational Biology and
  Bioinformatics}, 4(3):458--466, 2007.

\bibitem{bordewich}
M.~Bordewich and C.~Semple.
\newblock Computing the minimum number of hybridization events for a consistent
  evolutionary history.
\newblock {\em Discrete Applied Mathematics}, 155(8):914--928, 2007.

\bibitem{quantifyingreticulation}
J.~Collins, S.~Linz, and C.~Semple.
\newblock Quantifying hybridization in realistic time.
\newblock {\em Journal of Computational Biology}, 2010.
\newblock To appear.

\bibitem{downey1999}
R.~G. Downey and M.~R. Fellows.
\newblock {\em Parameterized Complexity}.
\newblock Springer-Verlag, 1999.

\bibitem{Flum2006}
J.~Flum and M.~Grohe.
\newblock {\em Parameterized Complexity Theory}.
\newblock Springer, 2006.

\bibitem{Gambette2009structure}
P.~Gambette, V.~Berry, and C.~Paul.
\newblock The structure of level-k phylogenetic networks.
\newblock In {\em Proceedings of the 20th Annual Symposium on Combinatorial
  Pattern Matching}, CPM '09, pages 289--300, Berlin, Heidelberg, 2009.
  Springer-Verlag.

\bibitem{MathEvPhyl}
O.~Gascuel, editor.
\newblock {\em Mathematics of Evolution and Phylogeny}.
\newblock Oxford University Press, Inc., 2005.

\bibitem{reconstructingevolution}
O.~Gascuel and M.~Steel, editors.
\newblock {\em Reconstructing Evolution: New Mathematical and Computational
  Advances}.
\newblock {Oxford University Press, USA}, 2007.

\bibitem{Gramm2008}
J.~Gramm, A.~Nickelsen, and T.~Tantau.
\newblock Fixed-parameter algorithms in phylogenetics.
\newblock {\em The Computer Journal}, 51(1):79--101, 2008.

\bibitem{gusfielddecomp2007}
D.~Gusfield, V.~Bansal, V.~Bafna, and Y.~Song.
\newblock A decomposition theory for phylogenetic networks and incompatible
  characters.
\newblock {\em Journal of Computational Biology}, 14(10):1247--1272, 2007.

\bibitem{gusfield2}
D.~Gusfield, D.~Hickerson, and S.~Eddhu.
\newblock An efficiently computed lower bound on the number of recombinations
  in phylognetic networks: Theory and empirical study.
\newblock {\em Discrete Applied Mathematics}, 155(6-7):806--830, 2007.

\bibitem{husonetalgalled2009}
D.~H. Huson, R.~Rupp, V.~Berry, P.~Gambette, and C.~Paul.
\newblock Computing galled networks from real data.
\newblock {\em Bioinformatics}, 25(12):i85--i93, 2009.

\bibitem{HusonRuppScornavacca10}
D.~H. Huson, R.~Rupp, and C.~Scornavacca.
\newblock {\em Phylogenetic Networks: Concepts, Algorithms and Applications}.
\newblock Cambridge University Press, 2011.
\newblock to appear.

\bibitem{surveycombinatorial2011}
D.~H. Huson and C~Scornavacca.
\newblock A survey of combinatorial methods for phylogenetic networks.
\newblock {\em Genome Biology and Evolution}, 3:23--35, 2011.

\bibitem{huynh}
T.N.D. Huynh, J.~Jansson, N.B. Nguyen, and W.-K. Sung.
\newblock Constructing a smallest refining galled phylogenetic network.
\newblock In {\em Research in Computational Molecular Biology (RECOMB)}, volume
  3500 of {\em Lecture Notes in Bioinformatics}, pages 265--280, 2005.

\bibitem{JanssonEtAl2006}
J.~Jansson, N.~B. Nguyen, and W-K. Sung.
\newblock Algorithms for combining rooted triplets into a galled phylogenetic
  network.
\newblock {\em SIAM Journal on Computing}, 35(5):1098--1121, 2006.

\bibitem{JanssonSung2006}
J.~Jansson and W-K. Sung.
\newblock Inferring a level-1 phylogenetic network from a dense set of rooted
  triplets.
\newblock {\em Theoretical Computer Science}, 363(1):60--68, 2006.

\bibitem{elusiveness}
Steven Kelk, Celine Scornavacca, and Leo van Iersel.
\newblock On the elusiveness of clusters, 2011.
\newblock Submitted to TCBB.

\bibitem{myers2003}
S.~R. Myers and R.~C. Griffiths.
\newblock Bounds on the minimum number of recombination events in a sample
  history.
\newblock {\em Genetics}, 163:375--394, 2003.

\bibitem{Nakhleh2009ProbSolv}
L.~Nakhleh.
\newblock {\em The Problem Solving Handbook for Computational Biology and
  Bioinformatics}, chapter Evolutionary phylogenetic networks: models and
  issues.
\newblock Springer, 2009.

\bibitem{niedermeier2006}
R.~Niedermeier.
\newblock {\em {Invitation to Fixed Parameter Algorithms (Oxford Lecture Series
  in Mathematics and Its Applications)}}.
\newblock {Oxford University Press, USA}, March 2006.

\bibitem{Semple2007}
C.~Semple.
\newblock {\em Reconstructing Evolution - New Mathematical and Computational
  Advances}, chapter Hybridization Networks.
\newblock Oxford University Press, 2007.

\bibitem{SempleSteel2003}
C.~Semple and M.~Steel.
\newblock {\em Phylogenetics}.
\newblock Oxford University Press, 2003.

\bibitem{tohabib2009}
T-H. To and M.~Habib.
\newblock Level-$k$ phylogenetic networks are constructable from a dense
  triplet set in polynomial time.
\newblock In {\em CPM09}, volume 5577 of {\em LNCS}, pages 275--288, 2009.

\bibitem{lev2TCBB}
L.~J.~J. van Iersel, J.~C.~M. Keijsper, S.~M. Kelk, L.~Stougie, F.~Hagen, and
  T.~Boekhout.
\newblock Constructing level-2 phylogenetic networks from triplets.
\newblock {\em IEEE/ACM Transactions on Computational Biology and
  Bioinformatics}, 6(4):667--681, 2009.

\bibitem{simplicityAlgorithmica}
L.~J.~J. van Iersel and S.~M. Kelk.
\newblock Constructing the simplest possible phylogenetic network from
  triplets.
\newblock {\em Algorithmica}, pages 1--29, 2009.
\newblock 10.1007/s00453-009-9333-0.

\bibitem{twotrees}
L.~J.~J. van Iersel and S.~M. Kelk.
\newblock When two trees go to war.
\newblock {\em Journal of Theoretical Biology}, 269(1):245--255, 2011.

\bibitem{reflections}
L.~J.~J. van Iersel, S.~M. Kelk, and M.~Mnich.
\newblock Uniqueness, intractability and exact algorithms: Reflections on
  level-$k$ phylogenetic networks.
\newblock {\em Journal of Bioinformatics and Computational Biology},
  7(2):597--623, 2009.

\bibitem{cass}
L.~J.~J. van Iersel, S.~M. Kelk, R.~Rupp, and D.~H. Huson.
\newblock Phylogenetic networks do not need to be complex: Using fewer
  reticulations to represent conflicting clusters.
\newblock {\em Bioinformatics}, 26:i124--i131, 2010.
\newblock Special issue: Proceedings of Intelligent Systems for Molecular
  Biology 2010 (ISMB2010), 10th-13th September 2010, Boston USA.

\bibitem{whiddenWABI}
C.~Whidden and N.~Zeh.
\newblock A unifying view on approximation and fpt of agreement forests.
\newblock In Steven Salzberg and Tandy Warnow, editors, {\em Algorithms in
  Bioinformatics}, volume 5724 of {\em Lecture Notes in Computer Science},
  pages 390--402. Springer Berlin / Heidelberg, 2009.

\bibitem{pirnISMB2010}
Y.~Wu.
\newblock Close lower and upper bounds for the minimum reticulate network of
  multiple phylogenetic trees.
\newblock {\em Bioinformatics}, 26:i140--i148, 2010.
\newblock Special issue: Proceedings of Intelligent Systems for Molecular
  Biology 2010 (ISMB2010), 10th-13th September 2010, Boston USA.

\bibitem{WuG08}
Y.~Wu and D.~Gusfield.
\newblock A new recombination lower bound and the minimum perfect phylogenetic
  forest problem.
\newblock {\em J. Comb. Optim.}, 16(3):229--247, 2008.

\end{thebibliography}

\end{document}